\documentclass{article}

\usepackage{microtype}
\usepackage{graphicx}
\usepackage{subcaption}
\usepackage{booktabs} 
\usepackage{float}    
\usepackage{placeins} 
\usepackage{wrapfig}
\usepackage{subcaption}
\usepackage{caption}
\usepackage[preprint]{Formatting_Instructions_For_NeurIPS_2026/neurips_2026}

\usepackage{amsmath}
\usepackage{amssymb}
\usepackage{amsfonts}
\usepackage{algorithm}
\usepackage{algorithmic}
\usepackage{mathtools}
\usepackage{amsthm}
\usepackage{booktabs}
\usepackage{multirow}
\usepackage{listings}
\usepackage{xcolor}
\usepackage{colortbl}
\usepackage{subcaption}
\usepackage{graphicx}
\usepackage[utf8]{inputenc}
\usepackage[T1]{fontenc}
\usepackage{hyperref}
\usepackage{url}
\usepackage{nicefrac}
\usepackage{float}    
\usepackage{makecell}
\usepackage{siunitx}
\usepackage[capitalize,noabbrev]{cleveref}

\theoremstyle{plain}
\newtheorem{theorem}{Theorem}[section]

\theoremstyle{definition}

\newtheorem{assumption}[theorem]{Assumption}
\theoremstyle{remark}
\newtheorem{remark}[theorem]{Remark}

\usepackage[textsize=tiny]{todonotes}

\newcommand{\methodname}{MAT-Cell}

\title{\methodname: A Multi-Agent Tree-Structured Reasoning Framework for Batch-Level Single-Cell Annotation}
\author{
  \makebox[\textwidth][c]{%
    \hspace*{-0.6em}%
    \begin{minipage}{0.95\textwidth}
      \centering
      Yehui Yang$^{1,2,3,4}$ \quad Zelin Zang$^{1,3}$ \quad Changxi Chi$^{1,3}$ \quad Jingbo Zhou$^{1}$ \quad Xienan Zheng$^{4}$ \quad Yuzhe Jia$^{4}$ \\
      Chang Yu$^{1}$ \quad Jinlin Wu$^{1}$ \quad Fuji Yang$^{3}$ \quad Jiebo Luo$^{3}$ \quad Zhen Lei$^{3}$ \quad Stan Z. Li$^{1}$ \\
      {\small $^{1}$Westlake University, Hangzhou, China} \\
      {\small $^{2}$Shenzhen University of Advanced Technology, Shenzhen, China} \\
      {\small $^{3}$Center for Artificial Intelligence and Robotics, Hong Kong Institute of Science and Innovation,} \\
      {\small Chinese Academy of Sciences, Hong Kong, China} \\
      {\small $^{4}$University Key Laboratory of Information and Communication Security Backup and Recovery,} \\
      {\small Yunnan Minzu University, Kunming, China} \\
      {\small \texttt{zangzelin@westlake.edu.cn}}
    \end{minipage}%
  }
}

\begin{document}
\maketitle

\begin{abstract}
Automated single-cell annotation is difficult when the most abundant genes are not the most discriminative ones, or when a target state is poorly covered by a fixed reference atlas.
GPTCelltype-style one-shot prompting allows large language models (LLMs) to produce plausible labels from generic expression signals, while reference-based annotators can force unfamiliar states into the nearest known category.
We propose \methodname{}, a prompt-driven framework for batch-level single-cell annotation that separates evidence grounding from label decision.
\methodname{} first uses Reverse Verification Query (RVQ) to combine tissue context, observed differentially expressed genes, and LLM-elicited biological priors into structured candidate-specific premises.
Verifier agents then convert these premises into explicit premise-to-claim reasoning trees, and bounded multi-round debate compares, challenges, and revises the resulting claims before consensus or final adjudication.
The returned Syllogistic Derivation Tree (SDT) provides an auditable debate trace rather than a formal proof of the annotation.
In open-candidate benchmarks across five datasets, a locally deployed Qwen3-30B model with \methodname{} achieves 75.5\% average accuracy, compared with 64.2\% for the strongest evaluated CoT baseline and 51.9\% for the strongest evaluated scPilot variant.
In oracle-candidate benchmarks across three species, \methodname{} remains competitive across backbones, and local inference substantially reduces monetary cost for batch annotation.
Code is available at:
\href{https://anonymous.4open.science/r/MATCell-4067}
{https://anonymous.4open.science/r/MATCell-4067}
\end{abstract}

\section{Introduction}
\label{intro}

\begin{figure}[t!]
    \setlength{\textfloatsep}{6pt plus 1pt minus 1pt}
    \centering
    \includegraphics[width=\linewidth]{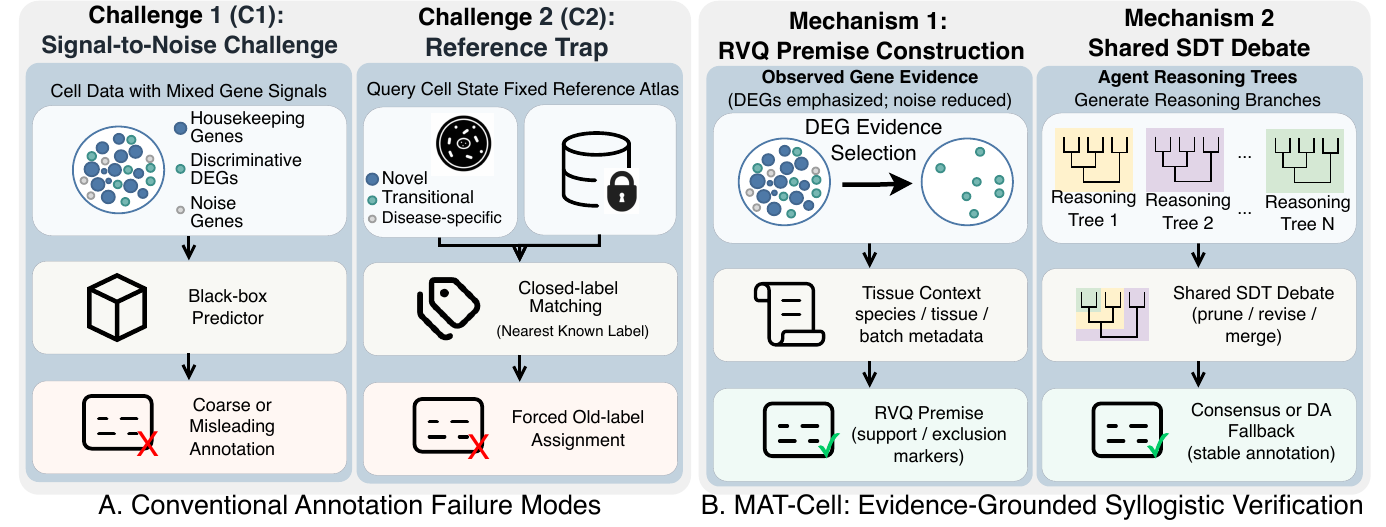}
    \vspace{-0.65em}
    \captionsetup{skip=0pt}
    \caption{\textbf{Cellular Reasoning Challenges and Solutions.} \textbf{(A) Challenges:} GPTCelltype-style annotation can be distracted by abundant but weakly discriminative genes, while reference-based systems can map unfamiliar states to the nearest known label. \textbf{(B) \methodname{}:} RVQ grounds the input in observed DEGs and tissue context. Each verifier agent then builds a premise-to-claim reasoning tree; multi-round debate compares these trees to reach consensus or final adjudication.}
    \label{fig:motivation}
    \vspace{-1.8em}
\end{figure}

Single-cell RNA sequencing (scRNA-seq)~\citep{lahnemann2020eleven, klein2015droplet} now profiles tissues at a scale where manual annotation is no longer practical~\citep{regev2017human, hao2024large}.
In a typical analysis, each cluster or batch must be assigned a cell identity from a short list of expressed genes, tissue context, and often incomplete prior knowledge.
This is difficult because the most visible genes are not always the most informative ones: abundant housekeeping, ribosomal, stress-response, or lineage-shared genes can dominate the input, while the genes that separate nearby subtypes may appear only in a smaller set of differentially expressed markers.
The problem is especially fragile for rare, transitional, disease-specific, or cross-species cell states, where a coarse label may look plausible even when the discriminative evidence points elsewhere~\citep{trapnell2015defining, wagner2016revealing}.

Reference-based annotation methods are reliable when the target population is well represented in an atlas.
Methods such as CellTypist~\citep{dominguez2022cross} and scANVI~\citep{xu2021probabilistic} compare new cells against learned reference manifolds and can scale to common cell types.
However, this design also creates a closed-label pressure: when a cluster lies outside the reference distribution, the method must still align it to the nearest known category.
As a result, transitional progenitors, disease-associated states, or fine-grained tissue-specific subtypes may be assigned confident but biologically coarse labels~\citep{luecken2022benchmarking, stuart2019comprehensive}.
The limitation is not only a lack of data coverage; it is also a lack of an explicit step that asks whether the proposed label is actually supported by the observed marker evidence.

Large language models (LLMs) offer a complementary possibility because they encode broad biological knowledge and can reason over open vocabularies~\citep{xiao2024cellagent, mao2025scagent, fang2025a}.
Yet GPTCelltype-style prompting is brittle for single-cell annotation.
When a prompt lists the top expressed genes, the model often explains the cluster using familiar but non-specific signals, producing a plausible cell type rather than checking whether the label is supported by discriminative DEGs.
Retrieval-Augmented Generation (RAG) can add external marker information, but static marker databases are incomplete for long-tail cell states, disease contexts, and cross-species settings.
Thus, an LLM-based annotator needs two safeguards: it should obtain biological priors beyond a fixed database, and it should be forced to compare those priors against the observed evidence rather than treating them as answers.

We propose \methodname{}, a prompt-driven framework for batch-level single-cell annotation built around three modules: RVQ premise construction, verifier-agent reasoning trees, and consensus/fallback assembly.
First, Reverse Verification Query (RVQ) combines tissue context, observed DEGs, and LLM-elicited biological priors into a structured premise for the current cluster or batch.
Second, each verifier agent turns this premise into a syllogistic reasoning tree: a biological major premise, an observed-evidence minor premise, and a candidate annotation claim.
These agent-level trees are written into a shared Syllogistic Derivation Tree (SDT)~\citep{smith1989prior, khemlani2012theories}, which we use as an auditable debate trace rather than a formal proof system.
Across rounds, agents compare claims, challenge unsupported premises, and revise disputed branches.
The final annotation is produced when the debate reaches a stable consensus, or by a Decision Agent for unresolved cases.
Together, these stages define a task-constrained inference-time protocol that gives cell-type annotation an explicit evidence chain, decision criterion, and fallback mechanism.

Our contributions are threefold.
\textbf{(1) Evidence grounding}: We introduce RVQ, which combines observed DEGs, tissue context, and elicited biological priors into structured premises instead of relying only on static marker retrieval.
\textbf{(2) Agent-level syllogistic reasoning}: We require each verifier agent to express its annotation as a premise-based reasoning tree rather than a single label.
\textbf{(3) Debate-based stabilization}: We aggregate and revise these trees through multi-round debate, producing a shared SDT and a consensus or adjudicated annotation across open-candidate and oracle-candidate benchmarks.

\section{Related Work}
\label{related}
\textbf{Reference-based annotation under distribution shift.}
Traditional automated cell type annotation methods primarily rely on supervised classification or latent space alignment against curated reference atlases.
Methods such as SingleR~\citep{aran2019reference}, CellTypist~\citep{dominguez2022cross}, and scANVI~\citep{xu2021probabilistic} formulate annotation as a statistical correlation problem within a closed manifold, enabling reliable identification of common cell states. However, they fundamentally operate as implicit pattern-matching approaches.
As a result, they suffer from severe \emph{reference bias and annotation distribution shift}: disease-specific subtypes or transitional states outside the reference manifold are often force-aligned to the nearest known cluster with high confidence.
Recent foundation models, including scGPT~\citep{scgpt2024}, Geneformer~\citep{theis2023geneformer}, and scFoundation~\citep{hao2024large}, scale annotation via Transformer architectures, yet encounter inherent \emph{expression sparsity and feature overlap}, where highly expressed housekeeping genes dominate attention and induce biologically plausible but incorrect annotations.
Critically, these models remain black-box probabilistic predictors and lack mechanisms for explicit biological or logical verification.

\textbf{LLM agents and biological reasoning.}
To address these limitations, recent works have explored agentic frameworks and reinforcement learning.
CellAgent~\citep{xiao2024cellagent}, scAgent~\citep{mao2025scagent}, and scPilot~\citep{gao2025scpilot} primarily orchestrate external bioinformatics tools or LLM agents, while Cell-o1~\citep{fang2025a} and CellDuality~\citep{anonymous2026cellduality} apply task-specific reinforcement learning; however, these paradigms lack explicit and generalizable reasoning traces.
In contrast, advances in LLM reasoning—including Chain-of-Thought~\citep{wei2022chain, kojima2022large}, Tree of Thoughts~\citep{yao2023tree}, Self-Consistency~\citep{wang2022self}, and multi-agent debate frameworks~\citep{liang2023encouraging, du2023improving, li2023camel}—demonstrate that structured reasoning and dialectic verification substantially reduce hallucinations.
\methodname{} bridges these paradigms to transcriptomic analysis by reformulating annotation as an \emph{evidence-grounded reasoning process}. By enforcing a structured logical template through the Reverse Verification Query (RVQ) and rule-guided multi-agent debate, \methodname{} enables transparent and generalizable annotation beyond static reference atlases.

\section{Methodology}
\label{method}

\setlength{\textfloatsep}{8pt plus 2pt minus 2pt}
\begin{figure*}[t!]
    \centering
    \centerline{\includegraphics[width=1\textwidth]{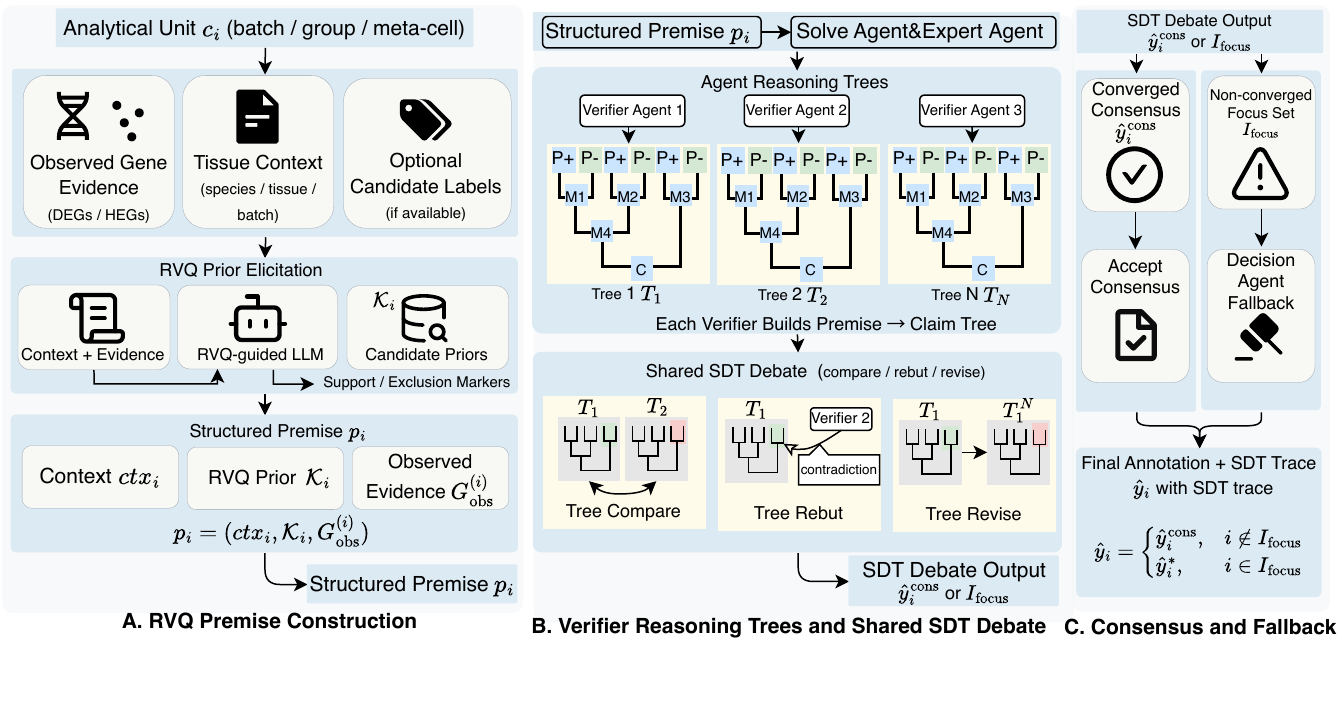}}
    \vspace{-0.85em}
    \captionsetup{skip=0pt}
    \caption{\textbf{The \methodname{} Framework.} For each analytical unit, \methodname{} takes observed gene evidence, tissue context, and optional candidate labels as input. RVQ queries an LLM for candidate priors and combines them with context and observed evidence into a structured premise. Verifier Agents build premise-to-claim reasoning trees, then compare, rebut, and revise them through a shared SDT. The final stage returns either a normalized exact-match consensus annotation or routes unresolved focus-set cases to a Decision Agent fallback, together with the SDT trace.}
    \label{fig:framework}
    \vspace{-1.1em}
\end{figure*}

\methodname{} does not ask an LLM to immediately name a cell type.
Instead, it decomposes annotation into three steps: elicit candidate-specific biological expectations, test these expectations against observed DEGs and context, and let multiple verifier agents build and challenge premise-to-claim reasoning trees in a shared SDT.
This design follows the broader observation that structured reasoning and debate can make LLM decisions more checkable~\citep{wei2022chain, yao2023tree, liang2023encouraging, du2023improving}.
Concretely, RVQ is the premise-construction step, optional inductive anchoring is an implementation-level candidate and marker filtering procedure when raw inputs require it, and the SDT is an auditable debate trace rather than a formal proof system; Appendix~\ref{appendix:method_details} gives the raw-matrix preprocessing details for this optional anchoring step.
For each batch-level analytical unit, the system receives observed genes, tissue or dataset context, and optionally a candidate label set.
It returns a cell-type annotation together with a structured debate trace.
Figure~\ref{fig:framework} gives the overall workflow, and the complete inference procedure is provided in Algorithm~\ref{alg:matcell} in Appendix~\ref{appendix:method_details}.

\subsection{Problem Setup: Batch-Level Annotation with Structured Traces}
Let $\mathbf{X} \in \mathbb{R}^{N \times G}$ denote a single-cell gene expression matrix with $N$ cells and $G$ genes.
\methodname{} operates on analytical units $\mathcal{C}=\{c_1,\dots,c_B\}$, where each $c_i$ may be a batch-level instance, a predefined group, or an unsupervised meta-cell cluster.
The goal is to predict a cell-type label $\hat{y}_i$ for each unit and to return a structured record of the reasoning process.

For each unit $c_i$, we extract observed evidence $\mathbf{e}_i=(\mathcal{G}^{(i)}_{\mathrm{DEG}},\mathcal{G}^{(i)}_{\mathrm{HEG}},\mathbf{ctx}_i)$, where $\mathcal{G}^{(i)}_{\mathrm{DEG}}$ denotes differentially expressed genes, $\mathcal{G}^{(i)}_{\mathrm{HEG}}$ denotes highly expressed genes, and $\mathbf{ctx}_i$ contains tissue, species, batch, or dataset metadata.
This follows the common use of marker-level evidence in cluster-level single-cell annotation systems~\citep{shao2020sccatch}.
In our experiments, these analytical units are the benchmark-provided batch-level instances or groups used by a recent open-candidate LLM-based single-cell annotation benchmark~\citep{gao2025scpilot}.
The meta-cell partition $\Phi:\mathbf{X}\to\mathcal{C}$ is a general extension for raw matrices where units are not pre-defined, rather than an additional clustering contribution evaluated in this work.
In the fixed-candidate setting, reasoning is restricted to a provided candidate set $\mathcal{Y}_{cand}$, following the oracle-candidate setup in Cell-o1~\citep{fang2025a}.
In the open-candidate setting, no candidate labels are given; a local candidate set is induced during inference from observed evidence, tissue context, and elicited biological priors.

\subsection{Reverse Verification Query for Evidence-Grounded Premises}
\label{sec:method_step1}
Reverse Verification Query (RVQ) constructs the premise that all downstream agents must use.
The key idea is to reverse the usual annotation order: instead of asking the model to immediately output a label, RVQ first asks what biological evidence should support each plausible label in the given tissue context.
For analytical unit $c_i$, RVQ conditions on $\mathcal{G}_{obs}^{(i)}=(\mathcal{G}^{(i)}_{\mathrm{DEG}},\mathcal{G}^{(i)}_{\mathrm{HEG}})$, $\mathbf{ctx}_i$, and, when available, $\mathcal{Y}_{cand}$. 

In the fixed-candidate setting, RVQ asks a source LLM to list expected supporting markers for each label in $\mathcal{Y}_{cand}$ under the given tissue context.
In the open-candidate setting, RVQ is not given candidate labels; it proposes plausible labels and associated marker evidence from $\mathcal{G}_{obs}^{(i)}$ and $\mathbf{ctx}_i$.
The RVQ source model can be the same backbone used by the downstream agents or a separate prior source; when relevant, experiments distinguish self-generated RVQ priors from stronger external priors.
The resulting unit-specific prior is
\begin{equation}
\mathcal{K}_i =
\left\{
\left(
\ell,
\mathcal{M}_{\ell}^{+}
\right)
\right\}_{\ell \in \mathcal{Y}^{prior}_i},
\end{equation}
where $\ell$ is a candidate label and $\mathcal{M}_{\ell}^{+}$ denotes expected supporting markers.
Because $\mathcal{K}_i$ is model-generated, \methodname{} treats it as reference hypotheses rather than ground truth.
Downstream agents must therefore verify every RVQ-supported candidate against the observed DEGs/HEGs and tissue context; a label suggested by $\mathcal{K}_i$ can be rejected if its expected markers are not supported by $\mathcal{G}_{obs}^{(i)}$.
The structured premise for unit $c_i$ is $\mathbf{p}_i=(\mathcal{G}_{obs}^{(i)},\mathcal{K}_i,\mathbf{ctx}_i)$.

\subsection{Agent-Level Syllogistic Tree Construction}
\label{sec:method_step2}
The second design choice is to make each agent produce a structured reasoning tree rather than a label alone.
Given $\mathbf{p}_i$, the Solve Agent (SA) first proposes a local candidate set and a seed reasoning node $z_i^{(0)}=SA(\mathbf{p}_i)=(M_i^{(0)},E_i^{(0)},\mathcal{Y}_i^{(0)})$, where $M_i^{(0)}$ summarizes candidate-level biological priors, $E_i^{(0)}$ summarizes observed evidence, and $\mathcal{Y}_i^{(0)}$ is the local candidate set.
We require $\mathcal{Y}_i^{(0)}\subseteq\mathcal{Y}_i^{space}$ and $|\mathcal{Y}_i^{(0)}|\ll|\mathcal{Y}_i^{space}|$, where $\mathcal{Y}_i^{space}$ is $\mathcal{Y}_{cand}$ in the fixed-candidate setting and an open vocabulary induced from $\mathcal{K}_i$, $\mathcal{G}_{obs}^{(i)}$, and $\mathbf{ctx}_i$ otherwise.

The Expert Agent (EA) assigns $R$ verifier roles $\{\pi_{i,1},\dots,\pi_{i,R}\}=EA(\mathbf{ctx}_i,\mathcal{Y}_i^{(0)};R)$ for unit $c_i$.
These roles correspond to complementary checks such as marker support, exclusion markers, lineage consistency, and tissue-context consistency.
At debate round $t$, each Verifier Agent (VA) reads the structured premise, its assigned role, and the current SDT state, then writes
\begin{equation}
z_{i,j}^{(t)}
=
VA_j(\mathbf{p}_i,\pi_{i,j},\mathcal{T}_i^{(t-1)})
=
(M_{i,j}^{(t)},E_{i,j}^{(t)},\hat{y}_{i,j}^{(t)}),
\end{equation}
where $M_{i,j}^{(t)}$ is the biological major premise, $E_{i,j}^{(t)}$ is the observed-evidence minor premise, and $\hat{y}_{i,j}^{(t)}$ is the candidate annotation claim.
Thus, each VA contributes one explicit premise-to-claim tree of the form ``biological expectation + observed evidence $\Rightarrow$ candidate label,'' rather than only a final answer.
In Fig.~\ref{fig:framework}, $P{+}$ and $P{-}$ denote supporting and exclusion premises, $M$ denotes marker evidence, and $C$ denotes the candidate claim.

\subsection{Shared SDT Debate for Consensus and Fallback}
\label{sec:method_step3}
The third design choice is to debate over agent-level trees in a shared SDT, drawing on syllogistic accounts of premise-based reasoning~\citep{smith1989prior, khemlani2012theories}.
After round $t$, the SDT for unit $c_i$ is updated as $\mathcal{T}_i^{(t)}=\mathcal{T}_i^{(t-1)}\cup\{z_{i,j}^{(t)}\}_{j=1}^{R}$.
The SDT is therefore not only a log of messages: it is the shared object on which agents compare support, expose contradictions, and revise candidate claims.
The shared SDT implements three operations shown in Fig.~\ref{fig:framework}: tree comparison, rebuttal of unsupported branches, and revision of disputed claims.
Each SDT branch records the candidate label, supporting marker evidence, exclusion or counterevidence when available, and a short rationale connecting the premise to the claim.
Let $A_i^{(t)}=\{\hat{y}_{i,1}^{(t)},\dots,\hat{y}_{i,R}^{(t)}\}$ be the set of verifier answers at round $t$.
Before testing consensus, labels are normalized by lowercasing, removing punctuation and extra whitespace, and applying the same benchmark-specific synonym or ontology mapping used for evaluation when such a mapping is available.
If the normalized answer set satisfies $|A_i^{(t)}|=1$ for any $t\le T_{\max}$, normalized exact-match consensus is reached and the final annotation is $\hat{y}_i^{cons}=\hat{y}_{i,1}^{(t)}$.
We use this conservative normalized exact-match criterion to avoid treating semantically related but biologically distinct labels as agreement in fine-grained annotation.
For batch processing, units that reach consensus are removed from the active focus set to avoid redundant computation.

Units still unresolved after $T_{\max}$ rounds are collected as $\mathcal{I}_{focus}=\{\,i \mid |A_i^{(T_{\max})}|>1\,\}$.
We refer to either $\hat{y}_i^{cons}$ or membership in $\mathcal{I}_{focus}$ as the SDT debate output passed to the final assembly stage.
For each $i\in\mathcal{I}_{focus}$, the Decision Agent receives the structured premise and final SDT state, producing $\hat{y}_i^{DA}=DA(\mathbf{p}_i,\mathcal{T}_i^{(T_{\max})})$.
The final prediction is
\begin{equation}
\hat{y}_i =
\begin{cases}
\hat{y}_i^{cons}, & i \notin \mathcal{I}_{focus}, \\
\hat{y}_i^{DA}, & i \in \mathcal{I}_{focus}.
\end{cases}
\end{equation}
The returned SDT $\hat{\mathcal{T}}_i$ is therefore an auditable debate trace: it records the initial candidate set, each verifier tree, revisions across rounds, and the final consensus or fallback decision.

\section{Experiments}
\label{exper}
\subsection{Experimental Setup}

We evaluate \methodname{} under two settings that separate candidate discovery from label selection.
In the \textit{open-candidate setting}, no candidate cell-type labels are given; the model must infer a candidate space from tissue context and marker evidence.
In the \textit{oracle-candidate setting}, the benchmark label set is given in advance, so the task is reduced to choosing among fixed candidates.
Open-candidate evaluation uses PBMC3K~\citep{10x_pbmc3k}, Liver~\citep{liang2022temporal}, Retina~\citep{menon2019single}, Brain, and Heart; the first three follow scPilot~\citep{gao2025scpilot}, while Brain and Heart are additional CELLxGENE-derived datasets used to test more heterogeneous tissue structures.
Oracle-candidate evaluation uses CellxGene-derived Human, Mouse, and Monkey benchmarks following Cell-o1~\citep{fang2025a}, with 2{,}400 batch-level instances per species.

Each analytical unit is represented by tissue or dataset context and a gene list.
In the open-candidate setting, \methodname{} uses DEGs as the main discriminative input.
In the oracle-candidate setting, we evaluate three gene-input views: \textit{Both} (HEGs + DEGs), \textit{DEG-only}, and \textit{HEG-only} (top-25 highly expressed genes).
HEGs test sensitivity to abundant but weakly discriminative genes, while DEGs provide the main discriminative evidence used by \methodname{}.
DEG extraction and context construction details are provided in Appendix~\ref{appendix:data_processing}.
Open-candidate results are reported as scPilot-style graded annotation scores~\citep{gao2025scpilot} with scores of 1, 0.5, and 0, macro-averaged over benchmark units; oracle-candidate results use exact string matching following Cell-o1~\citep{fang2025a}.
Baselines cover classical marker/signature methods, GPTCelltype prompting, CoT prompting, and recent LLM-based annotation systems, as listed chronologically in Tables~\ref{tab:oa_results} and~\ref{tab:ca_results}; detailed descriptions are provided in Appendix~\ref{appendix:evaluation_baselines}.
GPTCelltype asks the LLM to output a batch-level label in one shot, while CoT first elicits a short reasoning trace before the final label.
Unless otherwise specified, \methodname{} uses $R=5$ Verifier Agents (VAs), a maximum debate depth of $T_{\max}=3$, and temperature 0.7.
Unless otherwise specified, RVQ uses Gemini 3 as the default prior-source model.
For stochastic LLM-based methods, we report mean and standard deviation over independent runs.

\subsection{Main Results: Open-Candidate Annotation}

\begin{table*}[t]
\centering
\caption{\textbf{Open-candidate annotation results.}
Accuracy is the scPilot-style graded annotation score, reported as mean $\pm$ std (\%) across five benchmarks without oracle candidates.
Best and second-best values are shown in \textbf{bold} and \underline{underline}.
Red values indicate the relative change in average accuracy against the corresponding strongest non-MAT-Cell baseline; for unmatched backbones, the strongest non-MAT-Cell average is used.}
\resizebox{0.99\textwidth}{!}{
\label{tab:oa_results}
\newcommand{\pmstd}[1]{{\scriptsize$\pm$#1}}
\newcommand{\matgain}[1]{\textcolor{red}{{\scriptsize($\uparrow$ +#1\%)}}}
\newcommand{\matdrop}[1]{\textcolor{red}{{\scriptsize($\downarrow$ #1\%)}}}
\renewcommand{\arraystretch}{0.9}
\begin{tabular}{@{}l l | c c c c c | c@{}}
\toprule
\textbf{Method} & \textbf{Model} & \textbf{PBMC3K} & \textbf{Liver} & \textbf{Retina} & \textbf{Brain} & \textbf{Heart} & \textbf{Avg Acc.} \\
\midrule

CellTypist (Sci'22) & Non-LLM
& 46.4 & 56.3 & 38.8 & 24.2 & 69.0 & 46.9 \\
CellMarker2.0 (NAR'23) & Non-LLM
& 30.4 & 25.0 & 63.2 & 62.5 & 26.7 & 41.6 \\

\midrule
\multirow{5}{*}{\makecell[l]{CoT\\(NeurIPS'22)}} & Gemini 2.5 Pro
& 62.5\pmstd{4.4} & 78.2\pmstd{0.8} & 49.5\pmstd{1.2} & 66.9\pmstd{1.8} & 64.0\pmstd{2.8} & 64.2 \\
 & GPT-o1
& 67.5\pmstd{5.2} & 70.4\pmstd{1.3} & 49.0\pmstd{1.4} & 66.9\pmstd{4.2} & 61.3\pmstd{1.8} & 63.0 \\
 & Qwen3-14B
& 57.8\pmstd{3.1} & 64.8\pmstd{8.5} & 66.3\pmstd{7.3} & 64.0\pmstd{6.0} & 51.1\pmstd{10.5} & 60.8 \\
 & GPT-4o
& 68.8\pmstd{0.0} & \underline{80.4\pmstd{2.5}} & 44.2\pmstd{2.2} & 45.5\pmstd{2.0} & 58.0\pmstd{3.8} & 59.4 \\
 & Qwen3-30B
& 57.5\pmstd{2.5} & 78.1\pmstd{4.4} & 55.0\pmstd{17.8} & 57.7\pmstd{2.5} & 45.8\pmstd{4.3} & 58.8 \\

\midrule
\multirow{5}{*}{\makecell[l]{GPTCelltype\\(Nat Methods'24)}} & Qwen3-30B
& 67.5\pmstd{6.1} & 44.4\pmstd{6.2} & 74.7\pmstd{2.1} & 23.1\pmstd{9.7} & 24.6\pmstd{12.3} & 46.9 \\
 & GPT-o1
& 66.7\pmstd{0.5} & 56.0\pmstd{0.1} & 47.4\pmstd{0.2} & 29.6\pmstd{2.6} & 33.8\pmstd{9.2} & 46.7 \\
 & Qwen3-14B
& 67.4\pmstd{1.8} & 46.7\pmstd{9.0} & 74.7\pmstd{7.7} & 11.9\pmstd{9.8} & 12.3\pmstd{3.8} & 42.6 \\
 & GPT-4o
& 60.4\pmstd{0.5} & 44.0\pmstd{0.2} & 43.9\pmstd{0.2} & 35.6\pmstd{1.8} & 23.1\pmstd{16.9} & 41.4 \\
 & Gemini 2.5 Pro
& 58.3\pmstd{0.1} & 49.4\pmstd{0.7} & 49.1\pmstd{0.0} & 21.5\pmstd{8.2} & 18.5\pmstd{6.2} & 39.4 \\

\midrule
\multirow{5}{*}{\makecell[l]{scPilot\\(NeurIPS'25)}} & GPT-4o
& 64.6\pmstd{1.7} & 51.2\pmstd{0.2} & 67.5\pmstd{1.1} & 45.2\pmstd{22.6} & 30.8\pmstd{6.9} & 51.9 \\
 & GPT-o1
& \underline{79.2\pmstd{0.5}} & 51.8\pmstd{0.1} & 72.8\pmstd{0.7} & 11.5\pmstd{23.0} & 35.4\pmstd{15.8} & 50.1 \\
 & Qwen3-30B
& 75.0\pmstd{13.7} & 43.7\pmstd{5.4} & 73.7\pmstd{5.8} & 21.5\pmstd{10.2} & 20.0\pmstd{6.2} & 46.8 \\
 & Gemini 2.5 Pro
& 70.8\pmstd{2.1} & 48.8\pmstd{0.1} & \textbf{76.3\pmstd{0.0}} & 14.8\pmstd{5.2} & 16.9\pmstd{11.3} & 45.5 \\
 & Qwen3-14B
& 72.5\pmstd{5.0} & 42.2\pmstd{5.5} & \underline{74.7\pmstd{2.1}} & 2.2\pmstd{1.8} & 12.3\pmstd{3.8} & 40.8 \\

\midrule
\multirow{5}{*}{\makecell[l]{MAT-Cell\\(ours)}} & Qwen3-30B-c
& 61.3\pmstd{2.8} & 79.6\pmstd{1.9} & 51.0\pmstd{2.5} & 44.4\pmstd{1.4} & 60.0\pmstd{0.0} & 59.3 \matgain{0.9} \\
 & DeepSeek-V3
& 63.8\pmstd{2.8} & 68.5\pmstd{1.9} & 52.1\pmstd{2.9} & \underline{76.3\pmstd{6.5}} & 62.0\pmstd{1.8} & 64.5 \matgain{0.5} \\
 & Gemini 3
& 75.0\pmstd{0.0} & 80.4\pmstd{2.1} & 50.5\pmstd{3.5} & 65.6\pmstd{3.9} & \underline{78.0\pmstd{1.9}} & 69.9 \matgain{8.9} \\
 & GPT-4o
& 63.8\pmstd{5.2} & 75.5\pmstd{3.9} & 65.5\pmstd{7.4} & \textbf{79.4\pmstd{1.7}} & 73.3\pmstd{0.0} & \underline{71.5} \matgain{20.4} \\
 & Qwen3-30B
& \textbf{80.0\pmstd{2.8}} & \textbf{81.1\pmstd{4.8}} & 63.2\pmstd{1.9} & 71.9\pmstd{6.2} & \textbf{81.3\pmstd{3.8}} & \textbf{75.5} \matgain{28.4} \\

\bottomrule
\end{tabular}
}
\vspace{-1.2em}
\end{table*}

\textbf{Open-candidate gains reflect verified priors.}
Open-candidate annotation is difficult because the model must decide what labels are even plausible before it can choose one.
\methodname{} with Qwen3-30B reaches the highest average graded accuracy among the evaluated open-candidate methods in Table~\ref{tab:oa_results} at 75.5\%, compared with 64.2\% for the strongest CoT baseline and 51.9\% for the strongest scPilot variant.
This comparison emphasizes whether structured verification can improve a weaker, locally deployable backbone enough to compete with stronger API-based prompting baselines, rather than further increasing sampling budget on the strongest models.
Under the same Qwen3-30B backbone, the self-prior variant reaches 59.3\%, modestly above the Qwen3-30B CoT baseline at 58.8\%.
We treat prior-source quality as an explicit factor in the open-candidate setting: the Qwen3-30B self-prior variant tests the fully local workflow, while the Gemini-prior variant tests verification under higher candidate recall.
Thus, the gains reflect the interaction between candidate recall and evidence verification, rather than direct adoption of a stronger prior's output.
It also suggests that open-candidate performance is limited first by candidate recall and then by evidence verification: RVQ determines whether the correct label enters the local search space, while SDT debate determines whether that label survives marker-level checking.

\textbf{Fine-grained tissues reveal both gains and limits.}
The clearest gains appear in tissues where the correct label depends on context or fine-grained markers.
PBMC3K is comparatively easier, since common immune-cell labels are familiar to most methods.
Liver, Brain, and Heart expose the failure mode that motivates MAT-Cell: GPTCelltype-style prompting can settle on a nearby coarse label, while classical reference- or marker-based methods are constrained by the labels already represented in their atlas or marker database.
These cases are where the SDT is most useful: competing branches can challenge coarse or weakly supported labels with marker-level counterevidence before fallback adjudication.
The results also show a limit of the framework.
MAT-Cell is not best in every column, and its performance still depends on the backbone and on the quality of the biological priors.
The value of RVQ and SDT debate is therefore not that they make the model independent of evidence, but that they reduce reliance on a single uncontrolled generation step.

\vspace{-1.1em}
\subsection{Controlled Results: Oracle-Candidate Label Selection}

\textbf{Oracle candidates isolate label selection.}
\begin{table*}[t]
\centering
\caption{\textbf{Oracle-candidate cross-species results.}
Accuracy is mean $\pm$ std (\%) across Human, Mouse, and Monkey under \textit{Both}, \textit{DEG-only}, and \textit{HEG-only} views with oracle candidate labels.
Best and second-best values are shown in \textbf{bold} and \underline{underline}.
Red values report relative gains in Avg Acc. over the same-backbone CoT baseline; ``-c'' denotes self-generated RVQ priors.}
\label{tab:ca_results}
\setlength{\tabcolsep}{2pt}
\renewcommand{\arraystretch}{1.15}
\resizebox{0.98\textwidth}{!}{%
\newcommand{\pmstd}[1]{{\scriptsize$\pm$#1}}
\newcommand{\matgain}[1]{\textcolor{red}{{\scriptsize($\uparrow$ +#1\%)}}}
\begin{tabular}{l l | c c c c c c c c c | c}
\toprule
\multirow{2}{*}{\textbf{Method}} & \multirow{2}{*}{\textbf{Model}} &
\multicolumn{3}{c}{\textbf{Human}} &
\multicolumn{3}{c}{\textbf{Mouse}} &
\multicolumn{3}{c}{\textbf{Monkey}} & \multirow{2}{*}{\textbf{Avg Acc.}} \\
\cmidrule(lr){3-5} \cmidrule(lr){6-8} \cmidrule(lr){9-11}
& & Both & DEG & HEG & Both & DEG & HEG & Both & DEG & HEG & \\
\midrule
scCATCH (iScience'20) & Non-LLM
& 33.9\pmstd{1.0} & 30.2\pmstd{1.2} & 28.6\pmstd{1.4}
& 37.4\pmstd{0.9} & 36.7\pmstd{0.5} & 30.3\pmstd{1.0}
& 41.8\pmstd{0.7} & 42.2\pmstd{0.8} & 20.0\pmstd{0.8}  & 33.5 \\
Cell-ID (Nat Biotech'21) & Non-LLM
& 43.1\pmstd{0.4} & 43.2\pmstd{0.3} & 21.8\pmstd{1.1}
& 28.8\pmstd{1.3} & 28.9\pmstd{1.1} & 16.0\pmstd{0.7}
& 51.7\pmstd{0.5} & 49.8\pmstd{1.1} & 30.3\pmstd{0.7}  & 34.8 \\
Cell-o1 (arXiv'25) & Qwen2.5-7B 
& 42.9\pmstd{1.2} & 40.9\pmstd{1.5} & 24.3\pmstd{1.1}
& 39.0\pmstd{1.6} & 39.4\pmstd{1.7} & 23.2\pmstd{0.8}
& 69.5\pmstd{1.8} & 68.5\pmstd{1.9} & 50.3\pmstd{1.2}  & 44.2 \\

\midrule
\multirow{6}{*}{\makecell[l]{CoT\\(NeurIPS'22)}}
& Qwen3-14B
& 48.7\pmstd{0.5} & 61.8\pmstd{1.7} & 21.6\pmstd{0.4}
& 47.9\pmstd{0.8} & 53.4\pmstd{1.5} & 27.4\pmstd{0.9}
& 67.2\pmstd{0.8} & 72.0\pmstd{0.6} & 40.4\pmstd{1.0}  & 48.9 \\
& Qwen3-30B
& 55.5\pmstd{0.3} & 62.7\pmstd{0.5} & 21.3\pmstd{0.4}
& 47.6\pmstd{2.2} & 52.1\pmstd{0.7} & 25.6\pmstd{0.4}
& 64.9\pmstd{1.1} & 71.6\pmstd{1.1} & 35.7\pmstd{0.2}  & 48.6 \\
& Llama-3.1-70B
& 29.0\pmstd{1.5} & 48.6\pmstd{0.9} & 11.7\pmstd{0.5}
& 32.1\pmstd{0.7} & 42.8\pmstd{1.2} & 14.6\pmstd{0.3}
& 61.9\pmstd{1.2} & 65.1\pmstd{0.9} & 23.5\pmstd{0.9}  & 36.6 \\
& DeepSeek-V3
& 39.5\pmstd{21.5} & 74.5\pmstd{1.4} & 23.4\pmstd{0.4}
& 65.0\pmstd{0.4} & 62.6\pmstd{6.3} & 37.3\pmstd{0.4}
& 85.5\pmstd{0.3} & 85.5\pmstd{0.1} & 53.7\pmstd{0.5}  & 58.6 \\
& Gemini 2.5 Flash
& 64.9\pmstd{0.3} & 71.9\pmstd{0.8} & 23.8\pmstd{0.5}
& 59.4\pmstd{0.3} & 62.9\pmstd{0.1} & 34.1\pmstd{0.1}
& 83.4\pmstd{0.3} & 84.5\pmstd{0.5} & 49.8\pmstd{0.8}  & 59.4 \\
& GPT-4.1
& 67.2\pmstd{1.0} & 74.5\pmstd{0.5} & 17.6\pmstd{0.6}
& 66.3\pmstd{0.4} & 68.1\pmstd{0.8} & 33.2\pmstd{0.5}
& \underline{86.7\pmstd{1.0}} & \underline{86.8\pmstd{0.3}} & 52.0\pmstd{0.5}  & 61.4 \\

\midrule
\multirow{6}{*}{\makecell[l]{GPTCelltype\\(Nat Methods'24)}}
& Qwen3-14B
& 29.1\pmstd{0.8} & 40.9\pmstd{0.9} & 16.4\pmstd{0.4}
& 30.8\pmstd{1.0} & 34.8\pmstd{1.0} & 16.2\pmstd{0.4}
& 54.2\pmstd{0.5} & 58.0\pmstd{0.6} & 30.2\pmstd{0.6}  & 34.5 \\
& Qwen3-30B
& 34.0\pmstd{0.8} & 45.0\pmstd{1.3} & 16.5\pmstd{0.8}
& 34.7\pmstd{1.0} & 38.7\pmstd{0.8} & 17.8\pmstd{0.5}
& 58.8\pmstd{1.0} & 64.4\pmstd{0.9} & 33.1\pmstd{1.0}  & 38.1 \\
& Llama-3.1-70B
& 19.1\pmstd{0.6} & 28.3\pmstd{0.9} & 11.7\pmstd{0.3}
& 20.0\pmstd{0.8} & 29.0\pmstd{1.7} & 12.6\pmstd{0.7}
& 37.6\pmstd{0.5} & 47.1\pmstd{0.8} & 24.3\pmstd{0.7}  & 25.5 \\
& DeepSeek-V3
& 55.1\pmstd{0.7} & 63.2\pmstd{1.2} & 18.9\pmstd{0.5}
& 54.1\pmstd{1.5} & 56.7\pmstd{1.0} & 25.7\pmstd{0.5}
& 37.6\pmstd{0.5} & 47.1\pmstd{0.8} & 24.3\pmstd{0.7}  & 42.5 \\
& Gemini 2.5 Flash
& 64.1\pmstd{0.8} & 70.9\pmstd{0.9} & 18.6\pmstd{0.6}
& 63.4\pmstd{1.2} & 65.9\pmstd{1.5} & 30.6\pmstd{1.1}
& 85.9\pmstd{0.9} & 85.9\pmstd{0.9} & 52.7\pmstd{0.9}  & 59.8 \\
& GPT-4.1
& 66.3\pmstd{1.5} & 73.3\pmstd{0.7} & 20.4\pmstd{0.8}
& 65.5\pmstd{0.8} & 64.9\pmstd{1.3} & 33.5\pmstd{0.5}
& 85.3\pmstd{0.0} & 86.4\pmstd{0.5} & 55.6\pmstd{1.5}  & 61.2 \\

\midrule
\multirow{4}{*}{\makecell[l]{MAT-Cell\\(ours)}}
& Qwen3-30B-c
& 51.8\pmstd{0.6} & 62.6\pmstd{1.4} & 20.3\pmstd{1.2}
& 48.7\pmstd{0.3} & 55.9\pmstd{1.4} & 24.5\pmstd{0.5}
& 73.3\pmstd{0.7} & 74.5\pmstd{0.9} & 41.4\pmstd{0.0}  & 50.3 \matgain{3.5} \\
& Qwen3-14B
& 63.3\pmstd{0.7} & 75.2\pmstd{0.2} & 26.1\pmstd{0.1}
& \underline{80.0\pmstd{0.2}} & \textbf{82.6\pmstd{0.7}} & \textbf{50.9\pmstd{0.6}}
& 62.4\pmstd{0.4} & 68.7\pmstd{0.5} & 36.3\pmstd{0.9}  & 60.6 \matgain{23.9} \\
& Qwen3-30B
& \underline{66.4\pmstd{0.5}} & \underline{76.4\pmstd{1.1}} & \underline{28.2\pmstd{0.9}}
& \textbf{80.8\pmstd{0.3}} & \underline{82.5\pmstd{0.7}} & \underline{49.9\pmstd{0.7}}
& 70.0\pmstd{0.7} & 75.9\pmstd{0.6} & 39.5\pmstd{0.9}  & \underline{63.3} \matgain{30.2} \\
& Gemini 2.5 Flash
& \textbf{79.6\pmstd{0.9}} & \textbf{81.4\pmstd{0.3}} & \textbf{29.5\pmstd{1.0}}
& 74.6\pmstd{0.6} & 76.7\pmstd{0.8} & 40.5\pmstd{0.7}
& \textbf{88.6\pmstd{0.3}} & \textbf{87.7\pmstd{0.4}} & \textbf{57.5\pmstd{2.0}}  & \textbf{68.5} \matgain{15.3} \\
\bottomrule
\end{tabular}
}
\vspace{-0.8em}
\end{table*}

Because the candidate label set is fixed, this setting tests whether each method can use observed genes and tissue context to choose among known labels.
\methodname{} variants occupy the best or second-best positions on all Human columns, are particularly strong on Mouse with Qwen backbones, and remain competitive on Monkey, where GPT-4.1 and Gemini 2.5 Flash baselines are already strong.
Fixing candidates removes one source of open-vocabulary error, but it does not make annotation trivial: methods still need to support or reject each candidate using species- and tissue-specific evidence.
The remaining gains therefore reflect the verification stage more directly than in the open-candidate setting, while the smaller gaps on Monkey indicate a boundary condition: when the candidate set is fixed and the backbone already has strong species-specific knowledge, structured debate has less room to improve the final choice.

\textbf{DEGs shift the bottleneck to evidence quality.}
This pattern appears for GPTCelltype, CoT, and MAT-Cell, so it is a property of the evidence rather than a special advantage of our framework.
Adding HEGs to DEGs does not consistently help; in several cases, the \textit{Both} view is lower than DEG-only, suggesting that abundant genes can dilute discriminative evidence.
MAT-Cell improves label selection when discriminative DEGs are available, but it does not remove dependence on species coverage or backbone knowledge.
Monkey results are closer across methods, especially for GPT-4.1, Gemini 2.5 Flash, and MAT-Cell-Gemini 2.5 Flash.
Thus, once candidate labels are fixed, the main bottleneck shifts from candidate discovery to the quality of evidence used to support or reject each label.
Notably, Gemini 2.5 Flash is strongest on Human and Monkey but relatively weaker on Mouse, suggesting that performance still depends on the backbone's implicit biomedical world knowledge even when candidate labels are fixed.
This pattern may reflect greater exposure of general-purpose LLMs to human atlas and clinical knowledge, whereas mouse cell-type nomenclature and marker-gene systems are more species-specific~\citep{rood2025human,yao2023high}.

\subsection{Ablation and Sensitivity Analysis}

We next isolate three sources of performance: discriminative evidence, multi-agent verification, and the scale of the debate.

\begin{wraptable}{r}{0.52\textwidth}
    \centering
    \vspace{-0.35cm}
    \renewcommand{\arraystretch}{0.9}
    \resizebox{\linewidth}{!}{
    \begin{tabular}{c l | c c c | c c}
    \toprule
    \textbf{ID} & \textbf{Config} & \textbf{H} & \textbf{M} & \textbf{Mk} & \textbf{Gap}~$\downarrow$ & \textbf{RI}~$\uparrow$ \\
    \midrule
    M0 & \makecell[l]{GPTCelltype+HEG} & 16.5\% & 17.8\% & 33.1\% & 16.6\% & 73\% \\
    M1 & \makecell[l]{GPTCelltype+DEG} & 45.0\% & 38.7\% & 64.4\% & 25.7\% & 78\% \\
    M2 & \makecell[l]{CoT+HEG} & 21.3\% & 25.6\% & 35.7\% & 14.4\% & 75\% \\
    M3 & \makecell[l]{CoT+DEG} & 62.7\% & 52.1\% & 71.6\% & 19.5\% & 85\% \\
    M4 & \makecell[l]{MAT-Cell+HEG} & 28.2\% & 49.9\% & 39.5\% & 21.7\% & 80\% \\
    \midrule
    M5 & \makecell[l]{\textbf{MAT-Cell+DEG}} & \textbf{76.4\%} & \textbf{82.5\%} & \textbf{75.9\%} & \textbf{6.6\%} & \textbf{97\%} \\
    \bottomrule
    \end{tabular}
    }
    \captionsetup{font=footnotesize}
    \caption{\textbf{System-level ablation.} All variants use Qwen3-30B in the oracle-candidate setting. H/M/Mk denote Human, Mouse, and Monkey DEG-only accuracy; Gap is the max-min species gap, and RI is the relative invariance score defined in Appendix~\ref{appendix:ablation_extended}.}
    \label{tab:system_ablation}
    \vspace{-0.35cm}
\end{wraptable}

\textbf{Input Evidence Quality and Framework Effect.}
The ablation confirms that input quality is the first bottleneck.
HEG-only configurations remain weak across GPTCelltype, CoT, and MAT-Cell, while DEG-based inputs give a large jump.
This supports the motivation for RVQ grounding: the method should reason from discriminative evidence, not from the most abundant transcripts.
Beyond input quality, MAT-Cell improves over the GPTCelltype+DEG baseline and reduces the cross-species gap, suggesting that the debate stage contributes stability rather than simply adding another prompt.
The interaction is important: DEGs provide discriminative premises, and the SDT turns those premises into claims that can be challenged, revised, or rejected.
The weak HEG-only variants show the converse boundary: the framework cannot create discriminative evidence when the input genes mostly reflect abundance rather than cell-type specificity.

\textbf{Contribution of the Multi-Agent Reasoning Framework.}
The comparison between CoT+DEG and MAT-Cell+DEG isolates the value of multi-agent verification.
Single-agent reasoning can use DEGs, but it has no structured way to force competing explanations to confront counterevidence.
MAT-Cell adds that pressure through Verifier Agents and a shared SDT.

\begin{wrapfigure}{r}{0.43\textwidth}
    \centering
    \vspace{-0.3cm}
    \captionsetup{font=footnotesize}
    \includegraphics[width=\linewidth]{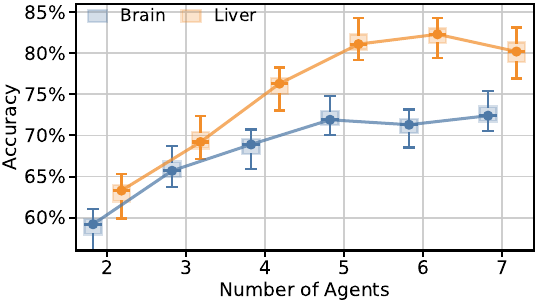}
    \caption{\textbf{Sensitivity to Council Scale ($R$).} Accuracy peaks at $R=5$ and then declines under excessive over-pruning.}
    \label{fig:agent_ablation}
    \vspace{0.4em}
    \renewcommand{\arraystretch}{1.05}
    \captionsetup{font=footnotesize}
    \captionof{table}{\textbf{Sensitivity to Dialectic Depth ($T_{\max}$).} Accuracy (\%) on Liver and Brain across rounds 2--4.}
    \label{tab:tmax_sensitivity}
    \resizebox{\linewidth}{!}{%
    \begin{tabular}{l | c c c}
    \toprule
    \textbf{Dataset} & $T_{\max}=2$ & $T_{\max}=3$ & $T_{\max}=4$ \\
    \midrule
    Liver & 70.4\% & 79.1\% & 66.7\% \\
    Brain & 54.6\% & 54.7\% & 56.0\% \\
    \bottomrule
    \end{tabular}}
    \vspace{-1cm}
\end{wrapfigure}

\textbf{Sensitivity to Council Scale and Dialectic Depth.}
More agents are useful only up to a point.
Accuracy improves as the council grows and peaks around $R=5$, after which excessive rebuttal can over-prune plausible branches.
Debate depth shows a similar trade-off: too few rounds leave weak claims insufficiently checked, while too many rounds add opportunities for drift or over-correction.
We therefore use $R=5$ and $T_{\max}=3$ as a bounded verification setting, rather than assuming that larger councils or longer debates are always better.

\vspace{-0.7em}
\subsection{Representative Case and Error Analysis}

\textbf{SDT traces make individual decisions auditable.}
We use the SDT as an audit trace for individual predictions rather than as an additional aggregate metric.
Fig.~\ref{fig:sdt_case_error} illustrates a representative workflow, where green labels denote correct predictions and red labels denote incorrect predictions.
In Round~1, the Solve Agent first proposes candidate labels for each analytical unit.
Verifier Agents then conduct dialectic verification according to the templates produced by the Expert Agent, exposing uncertain cases such as \textit{NK cell} versus \textit{T cell}, \textit{beta-IC} versus \textit{loop thin}, and \textit{vasa recta descending limb cell} versus \textit{glomerular endothelial cell}.
Across later rounds, the SDT records rebuttals, revisions, and the biological evidence behind each update.
Agents re-evaluate candidate labels using marker combinations and tissue context, for example using endothelial markers to distinguish glomerular endothelial cells from vasa recta subtypes, or using TCR and cytotoxic-effector signals to resolve the NK/T boundary.
When consensus is still not reached, the Decision Agent receives the full SDT trace and performs fallback adjudication.
In the illustrated case, the retained evidence chain supports \textit{vasa recta descending limb cell} as the final decision.

\textbf{Unresolved cases reflect evidence ambiguity.}
We further categorize non-consensus cases that trigger fallback adjudication to diagnose where the workflow remains uncertain.
As shown in Fig.~\ref{fig:sdt_case_error}, 64\% of unresolved cases arise from residual evidence ambiguity after RVQ grounding, 22\% from persistent agent disagreement before final adjudication, and 14\% from weak or conflicting DEG evidence.
This analysis shows that \methodname{} reduces the instability of one-shot generative annotation by organizing candidate generation, evidence verification, rebuttal, and adjudication into an auditable workflow, while its unresolved cases are still constrained by marker overlap among fine-grained cell types and by input DEG quality.

\begin{figure*}[t]
    \centering
    \includegraphics[width=0.98\textwidth]{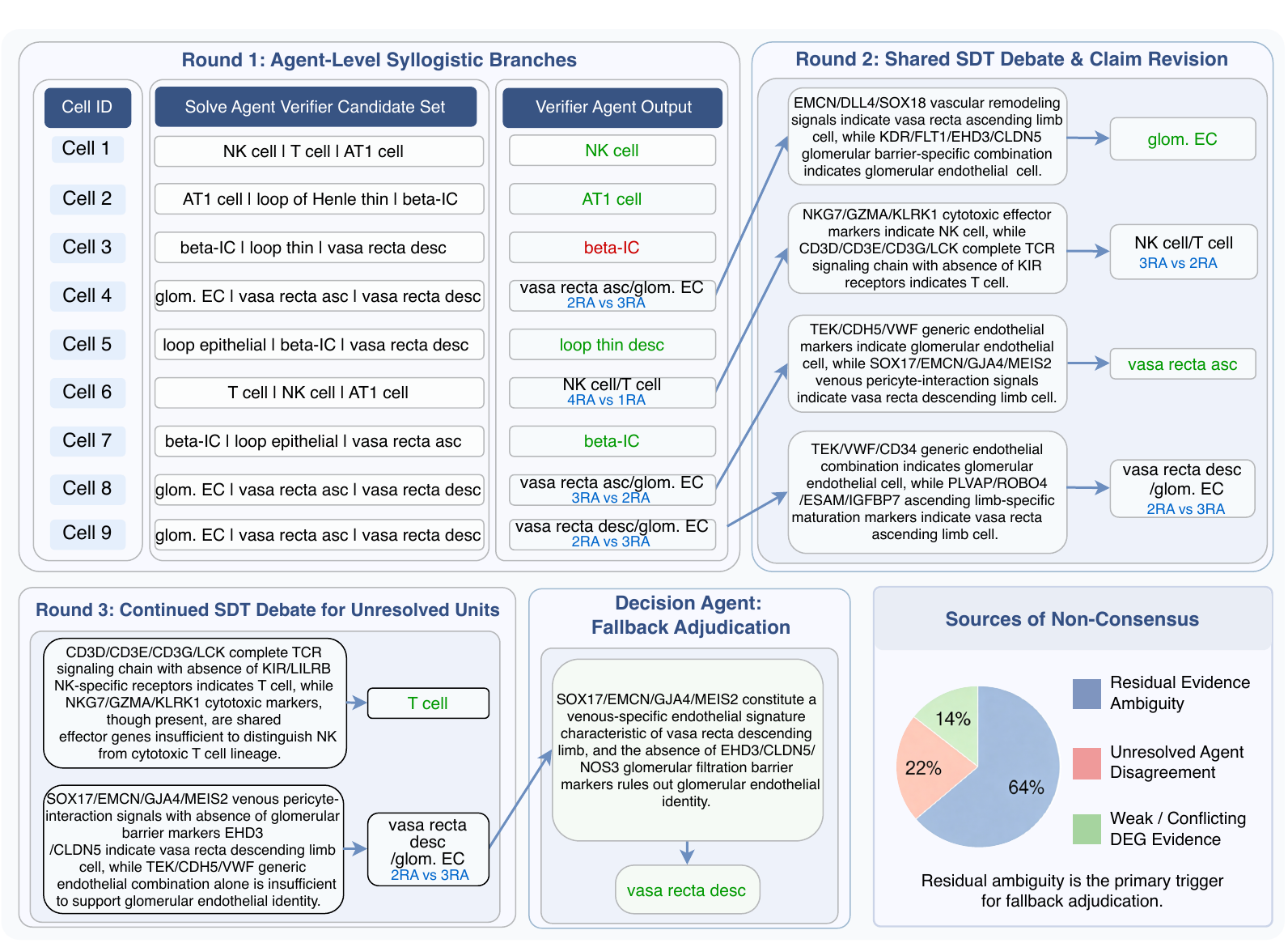} 
    \vspace{-0.75em}
    \caption{\textbf{Representative SDT workflow and non-consensus sources.}
    Green labels indicate correct predictions and red labels indicate incorrect predictions.
    Solve Verifier Agent candidates are examined by Expert Verifier Agent templates and Verifier Agent debate in the shared SDT; unresolved units are sent to the Decision Agent for fallback adjudication.
    The pie chart reports the main non-consensus sources that trigger fallback decisions.}
    \label{fig:sdt_case_error}
    \vspace{-0.55em}
\end{figure*}

\subsection{Performance, Cost, and Efficiency Analysis}
\vspace{-1.0em}
\leavevmode
\begin{wraptable}{r}{0.52\textwidth}
    \centering
    \captionsetup{font=footnotesize}
    \caption{\textbf{Accuracy, time, and monetary cost under API and local inference.}
    Accuracy is the Mouse, DEG-only column in Table~\ref{tab:ca_results}; time/cost is measured on the Human DEG efficiency setting.
    Local cost is amortized over 20-sample A100 batches.}
    \label{tab:efficiency}
    \resizebox{\linewidth}{!}{%
    \setlength{\tabcolsep}{3pt}
    \begin{tabular}{@{}l | cccc@{}}
    \toprule
    \textbf{Metric} & \textbf{CoT} & \textbf{CoT} & \textbf{GPTCelltype} & \textbf{MAT-Cell} \\
    \midrule
    Model & GPT-4.1 & Gemini 2.5 Flash & GPT-4.1 & Qwen3-30B \\
    Deployment & API & API & API & Local A100 \\
    Accuracy & 68.1\% & 62.9\% & 64.9\% & \textbf{82.5\%} \\
    Batch size & 1 & 1 & 1 & 20 \\
    Time/query & $\sim$8.0s & $\sim$6.2s & 1.93s & 9.17s \\
    Cost/query & $\sim$0.412 \textyen & $\sim$0.356 \textyen & $\sim$0.1100 \textyen & \textbf{$\sim$0.0079 \textyen} \\
    \bottomrule
    \end{tabular}
    }%
    \vspace{0.15em}
    \vspace{-0.35cm}
\end{wraptable}

\noindent Table~\ref{tab:efficiency} separates label-selection accuracy from deployment cost by pairing Mouse DEG-only accuracy with measured time and cost from the Human DEG efficiency setting.
On Mouse DEG-only annotation, \methodname{} reaches 82.5\%, compared with 68.1\% for GPT-4.1 CoT and 64.9\% for GPTCelltype.
MAT-Cell is slower than a single API call, but cell batches are independent and can be processed in parallel.
On one A100, a 20-sample batch takes 183.3 seconds, which gives an amortized latency of 9.17 seconds and a cost of about 0.0079 RMB per query.
The relevant trade-off is therefore not latency-free inference, but higher accuracy in this comparison with lower monetary cost when annotation is run as a local batch workload.
Thus, the efficiency claim is deployment-conditional: \methodname{} is most appropriate for offline or batched annotation workflows where throughput and cost dominate single-query latency.

\vspace{-1.0em}
\section{Conclusion}
\label{Conclusion}

\methodname{} introduces a prompt-driven structured reasoning framework that casts batch-level single-cell annotation as auditable, evidence-grounded decision making. RVQ constructs candidate-specific biological premises from observed genes and tissue context, while verifier agents organize their claims into Syllogistic Derivation Trees (SDTs) that expose premises, counterarguments, revisions, and final adjudication. Across open-candidate and oracle-candidate benchmarks, the results suggest that structured premise construction and bounded multi-agent debate can improve annotation reliability, while local inference offers a practical path to lower-cost batch annotation.

\bibliography{bib/example_paper}
\bibliographystyle{plainnat}



\appendix
\onecolumn

\setcounter{section}{0}
\renewcommand{\thesection}{\Alph{section}}

\begin{center}
\Large\bfseries Appendix Contents
\end{center}
\vspace{1em}

\makeatletter
\newif\if@inappendix
\@inappendixfalse

\let\oldcontentsline\contentsline
\renewcommand{\contentsline}[4]{%
  \if@inappendix
    \oldcontentsline{#1}{#2}{#3}{#4}%
  \fi
}
\makeatother

\addtocontents{toc}{\protect\makeatletter\protect\@inappendixtrue\protect\makeatother}

\makeatletter
\renewcommand\tableofcontents{%
    \@starttoc{toc}%
}
\makeatother

\tableofcontents
\newpage

\section{Theoretical Analysis and Proofs}
\label{app:theory}

In this section, we provide mathematical analyses for three theoretical properties of \methodname{}. We adopt a probabilistic framework to analyze the error bounds of the Dialectic Verification mechanism, the bounded-termination property of the Syllogistic Derivation Tree (SDT) construction procedure, and the asymptotic identifiability of novel cell states via Inductive Anchoring.

\subsection{Proof of Error Bound for Dialectic Verification (Theorem 1)}
\label{app:theory_consensus}

We model the Dialectic Verification process as a consensus problem among a committee of noisy binary classifiers.

\textbf{Setup.} Let $v$ be a proposed Logical Tuple with ground truth validity $y \in \{0, 1\}$. The Council of Verifiers consists of $R$ agents $\{f_{\text{ver}}^{(j)}\}_{j=1}^R$. Let $X_j = \mathbb{I}(f_{\text{ver}}^{(j)}(v) = \text{valid})$ be the binary indicator variable for the $j$-th agent's approval.

\begin{assumption}[Bounded Independent Error]
\label{assum:bounded_error}
We assume that the verifier agents are independent conditionally on the input tuple $v$, and each agent has a bounded error rate $\epsilon < 0.5$. Formally:
\begin{align}
    P(X_j = 1 \mid y=0) &\le \epsilon \quad \text{(False Positive Rate)} \\
    P(X_j = 0 \mid y=1) &\le \epsilon \quad \text{(False Negative Rate)}
\end{align}
\end{assumption}

\begin{remark}[Justification for Conditional Independence via EA]
In practice, standard LLM agents often exhibit highly correlated errors due to shared pre-training data. \methodname{} mitigates this risk by introducing the meta-level Expert Agent (EA). For each sample, the EA dynamically synthesizes $R$ exclusive prompt templates $\{\pi_1, \dots, \pi_R\}$ with differentiated inductive biases. By encouraging distinct expert perspectives (e.g., marker matching vs. lineage exclusion), the EA diversifies the reasoning paths of the Verifier Agents (VAs) and provides a practical mechanism supporting the conditional independence assumption.
\end{remark}

Recall the consensus criterion used in the main text: the council accepts a proposed tuple $v$ only when all verifier agents agree.
Equivalently, the system accepts $v$ if and only if $X_1 = X_2 = \cdots = X_R = 1$, i.e., every agent approves the tuple as valid.
This is a \emph{strict unanimous-consensus} rule designed to suppress hallucinated logical steps.

\begin{theorem}[Unanimous Consensus Suppresses Hallucinations]
\label{thm:app_error_bound}
For a false tuple ($y=0$), the probability that the council incorrectly accepts it (Type I Error / hallucination) under unanimous consensus satisfies
\begin{equation}
    P(\text{Accept} \mid y=0) \le \epsilon^R .
\end{equation}
For a true tuple ($y=1$), the probability that the council rejects it satisfies
\begin{equation}
    P(\text{Reject} \mid y=1) \le 1 - (1-\epsilon)^R .
\end{equation}
In particular, the hallucination probability decays exponentially in the council size $R$.
\end{theorem}

\begin{proof}
Under Assumption~\ref{assum:bounded_error}, for a false tuple ($y=0$) each agent approves with probability at most $\epsilon$, i.e., $P(X_j=1\mid y=0)\le \epsilon$.
Under the unanimous-consensus rule, the tuple is accepted only if all $R$ agents approve:
\begin{equation}
P(\text{Accept}\mid y=0)=P(X_1=\cdots=X_R=1\mid y=0)=\prod_{j=1}^R P(X_j=1\mid y=0)\le \epsilon^R,
\end{equation}
where the product form follows from conditional independence.

Similarly, for a true tuple ($y=1$), each agent rejects with probability at most $\epsilon$, i.e., $P(X_j=0\mid y=1)\le \epsilon$, hence $P(X_j=1\mid y=1)\ge 1-\epsilon$.
The probability that all agents approve is at least $(1-\epsilon)^R$, so the rejection probability is bounded by
\begin{equation}
P(\text{Reject}\mid y=1)=1-P(\text{Accept}\mid y=1)\le 1-(1-\epsilon)^R.
\end{equation}
This proves the stated bounds and shows exponential suppression of hallucinations as $R$ increases.
\end{proof}

\subsection{Proof of Bounded Termination for Syllogistic Derivation Tree (Theorem 2)}
\label{app:theory_convergence}

\textbf{Setup.} Let $\mathcal{Y}$ denote the finite ontology of candidate cell types.
Given the input card $\mathbf{x}_i^{card}$, the Solve Agent produces a candidate set $\mathcal{C}_i^{cand}\subseteq \mathcal{Y}$.
At dialectic round $t\in\{1,\dots,T_{\max}\}$, each verifier agent outputs a tentative conclusion
$\hat{y}_{i,j}^{(t)} = VA_j(\mathbf{x}_i^{card},\pi_j, \mathcal{T}_{i}^{(t-1)})$.
The SDT state $\mathcal{T}_{i}^{(t-1)}$ stores the council's hypotheses and rebuttals up to round $t-1$.

\begin{theorem}[Bounded-Termination of SDT Construction]
\label{thm:app_convergence}
The SDT construction procedure terminates in at most $T_{\max}$ dialectic rounds.
If unanimous consensus is reached at some round $t\le T_{\max}$, the algorithm outputs the consensus label $y_i$.
Otherwise, it falls back to the Decision Agent and outputs an adjudicated label based on the final evidence tree.
\end{theorem}

\begin{proof}
At each round $t$, the algorithm performs a finite council interaction and then checks the unanimous-consensus condition
$|A_i^{(t)}| = 1$, where $A_i^{(t)} = \{\hat{y}_{i,1}^{(t)},\dots,\hat{y}_{i,R}^{(t)}\}$.
If the condition holds, the procedure halts immediately and returns the consensus label.
If not, the procedure updates the tree state $\mathcal{T}_{i}^{(t)}$ and increments $t$.
Since $t$ is bounded by the predefined maximum depth $T_{\max}$, the loop can execute at most $T_{\max}$ times.
Therefore, the procedure must terminate either by consensus at some $t\le T_{\max}$ or by reaching $t=T_{\max}$, after which the Decision Agent is invoked.
\end{proof}

\subsection{Proof of OOD Identifiability (Theorem 3)}
\label{app:theory_ood}

\textbf{Setup.} Inductive Anchoring constructs a focused evidence space by intersecting observed cluster genes with elicited marker axioms.
Let $\mathcal{G}_{obs}$ denote the observed evidence genes for a query cluster (e.g., Top genes and/or DEGs), and let $\mathrm{Span}(\mathcal{K})$ denote the union of marker genes elicited for the candidate set.
Define the anchored evidence set $\mathcal{S}=\mathcal{G}_{obs}\cap \mathrm{Span}(\mathcal{K})$.
We analyze identifiability using a marker gene $g\in \mathcal{S}$, where $X_g$ is its expression in the query cluster and $Y_g$ is its expression in the background context.

\begin{assumption}[Gaussian Signal Model]
\label{assum:gaussian}
We assume gene expression levels (after log-normalization) follow Gaussian distributions:
\begin{itemize}
    \item Background: $Y_g \sim P_{\text{in}}(g) = \mathcal{N}(\mu_{\text{in}}, \sigma^2)$
    \item Novel State: $X_g \sim P_{\text{ood}}(g) = \mathcal{N}(\mu_{\text{ood}}, \sigma^2)$
\end{itemize}
The signal magnitude is defined as the shift $\Delta \mu = |\mu_{\text{ood}} - \mu_{\text{in}}|$.
\end{assumption}

We use the Contextual Divergence score $D_{\text{ctx}}(g)=|X_g-\mu_{\text{in}}|$ as a simple proxy for marker saliency under anchored evidence, and show it is statistically distinguishable from noise when the marker exhibits a mean shift.

\begin{theorem}[Asymptotic Separability]
\label{thm:app_ood}
For any error probability $\delta > 0$, there exists a signal-to-noise ratio threshold such that if $\frac{\Delta \mu}{\sigma} > \Phi^{-1}(1 - \delta/2) + \Phi^{-1}(1 - \delta)$, the Contextual Divergence score will identify the marker gene with probability at least $1-\delta$.
\end{theorem}

\begin{proof}
Let $Z = X_g - \mu_{\text{in}}$. Under Assumption~\ref{assum:gaussian}:
\begin{itemize}
    \item Under Null Hypothesis $H_0$ (Noise gene, $\mu_{\text{ood}} = \mu_{\text{in}}$): $Z \sim \mathcal{N}(0, \sigma^2)$.
    \item Under Alternative Hypothesis $H_1$ (Marker gene, $\mu_{\text{ood}} \neq \mu_{\text{in}}$): $Z \sim \mathcal{N}(\Delta \mu, \sigma^2)$ (assuming $\mu_{\text{ood}} > \mu_{\text{in}}$ w.l.o.g.).
\end{itemize}

The detection rule is $D_{\text{ctx}}(g) > \eta$, where $\eta$ is a critical value determined by the significance level $\alpha$ (False Positive Rate).
To control FPR at $\alpha$, we set $\eta$ such that $P(|Z| > \eta \mid H_0) = \alpha$.
Using the properties of the standard normal CDF $\Phi$:
\begin{equation}
    \eta = \sigma \cdot \Phi^{-1}(1 - \alpha/2)
\end{equation}

The Probability of Detection (Power) is $P(D_{\text{ctx}}(g) > \eta \mid H_1)$.
\begin{align}
    P(|Z| > \eta \mid H_1) &\ge P(Z > \eta \mid H_1) \quad (\text{since } Z \text{ is shifted positive}) \\
    &= P\left( \frac{Z - \Delta \mu}{\sigma} > \frac{\eta - \Delta \mu}{\sigma} \right) \\
    &= 1 - \Phi\left( \frac{\eta - \Delta \mu}{\sigma} \right) \\
    &= \Phi\left( \frac{\Delta \mu - \eta}{\sigma} \right)
\end{align}
We require the detection probability to be at least $1 - \beta$ (where $\beta$ is the Type II error rate). Let $\beta = \delta$ and $\alpha = \delta$.
\begin{equation}
    \Phi\left( \frac{\Delta \mu - \eta}{\sigma} \right) \ge 1 - \delta \implies \frac{\Delta \mu - \eta}{\sigma} \ge \Phi^{-1}(1 - \delta)
\end{equation}
Substituting $\eta$:
\begin{equation}
    \frac{\Delta \mu - \sigma \Phi^{-1}(1 - \delta/2)}{\sigma} \ge \Phi^{-1}(1 - \delta)
\end{equation}
\begin{equation}
    \frac{\Delta \mu}{\sigma} \ge \Phi^{-1}(1 - \delta/2) + \Phi^{-1}(1 - \delta)
\end{equation}
This inequality relates the signal-to-noise ratio (SNR) $\Delta\mu/\sigma$ to the target error level $\delta$.
It shows that, once Inductive Anchoring restricts reasoning to $\mathcal{S}=\mathcal{G}_{obs}\cap \mathrm{Span}(\mathcal{K})$, markers with sufficient mean shift are identified with high probability, whereas non-specific housekeeping noise outside $\mathrm{Span}(\mathcal{K})$ is excluded by construction.
Equivalently, in information-theoretic terms, the KL divergence $D_{\text{KL}}(P_{\text{ood}}\parallel P_{\text{in}})=(\Delta\mu)^2/(2\sigma^2)$ increases with the marker shift, yielding higher detection power and providing a principled basis for separating OOD states when discriminative markers exist in the retrieved span.
\end{proof}

\subsection{Necessity of Multi-Agent Reasoning}
\label{appendix:necessity}

To empirically validate the necessity of the multi-agent dialectic architecture over single-agent GPTCelltype-style prompting, we tracked internal state metrics during the inference process. The following data was collected on the Human dataset using the DEG view ($n=200$ batches).

\textbf{1. Substantial Gains in Cell-Level Annotation.}
First, the multi-agent architecture yields immediate and significant performance benefits. As shown in Part I of Table~\ref{tab:app_necessity}, relying solely on single-agent GPTCelltype prompting achieves a cell-level mean partial accuracy of only 0.4278. In contrast, by introducing multi-agent debate and tree-based reasoning, the partial accuracy of \methodname{} leaps to 0.7268. This demonstrates that precise annotation of cell types is a complex task requiring structured evidence auditing, where single-pass generation is highly susceptible to hallucinations.

\textbf{2. Debate Effectively Drives Consensus Convergence.}
To verify that the ``debate'' process genuinely resolves logical disagreements rather than merely casting blind votes, we analyzed the full agreement rate among the Verifier Agents (VAs). As reported in Part II of Table~\ref{tab:app_necessity}, prior to the debate (first VA round), the initial consensus rate is only 0.7632, largely because each agent is assigned a differentiated inductive bias by the Expert Agent (EA). However, after multiple rounds of structured dialectic verification and mutual evidence auditing on the shared tree state, the agreement rate at the final round (last VA round) surges to 0.9052. This marked increase in agreement suggests that the system mitigates many logical conflicts through mutual error correction. The high per-batch Pearson correlation ($r=0.7538$) further underscores the stability of this consensus-checking process.

\textbf{3. Decision Agent (DA) Effectively Resolves Boundary Cases.}
We further analyzed the fallback frequency of the final adjudication mechanism (Decision Agent, DA) and its performance on highly challenging samples. As detailed in Part III of Table~\ref{tab:app_necessity}, the vast majority of batches (88.50\%, $n=177$) successfully reach a natural consensus within the predefined maximum rounds, achieving a high partial accuracy of 0.7333.
Only 11.50\% of extremely difficult batches ($n=23$) fail to reach consensus during the multi-round debate, thereby triggering the DA fallback mechanism. It is worth noting that, although these samples represent complex boundary cases where even expert agents cannot reach a unanimous agreement, the DA manages to achieve a respectable accuracy of 0.6765 by reviewing the complete dispute branch of the Syllogistic Derivation Tree. This indicates that Contextual Synthesis and the DA's global adjudication mechanism provide a practical fallback, improving system robustness under extreme heterogeneity.

\begin{table}[ht]
\centering
\caption{\textbf{Empirical Metrics on the Necessity of Multi-Agent Reasoning.}
This table reports the internal state metrics from the Human (DEG view) dataset (total $n=200$ batches) to quantitatively demonstrate the necessity and efficacy of multi-agent debate and contextual synthesis.}
\label{tab:app_necessity}
\resizebox{0.9\textwidth}{!}{
\begin{tabular}{l l c}
\toprule
\textbf{Analytical Dimension} & \textbf{Metric Description} & \textbf{Value} \\
\midrule
\multirow{2}{*}{\textbf{Part I: Overall Performance Gains}}
& Cell-level Avg. Partial Accuracy (GPTCelltype baseline) & 0.4278 \\
& Cell-level Avg. Partial Accuracy (\methodname{})    & \textbf{0.7268} \\
\midrule
\multirow{3}{*}{\textbf{Part II: Debate Convergence}}
& Pre-debate (First Round) Full Agreement Rate (Mean) & 0.7632 \\
& Post-debate (Last Round) Full Agreement Rate (Mean) & \textbf{0.9052} \\
& Pearson $r$ (Pre vs. Post Agreement, per batch)      & 0.7538 \\
\midrule
\multirow{3}{*}{\textbf{Part III: Decision Agent (DA) Efficacy}}
& DA Fallback Frequency (Disputed Batches)            & 11.50\% ($n=23$) \\
& Avg. Partial Accuracy on Consensus Batches (No DA)  & 0.7333 ($n=177$) \\
& Avg. Partial Accuracy on Fallback Batches (by DA)   & \textbf{0.6765} ($n=23$) \\
\bottomrule
\end{tabular}
}
\end{table}

\newpage

\section{Dataset Documentation}
\label{appendix:dataset}

\subsection{Evaluation Baselines}
\label{appendix:evaluation_baselines}

The main experiments compare \methodname{} against three groups of baselines.
First, classical marker- or signature-based annotators include scCATCH~\citep{shao2020sccatch}, Cell-ID~\citep{cortal2021cellid}, CellTypist~\citep{dominguez2022cross}, and CellMarker2.0~\citep{hu2023cellmarker2}.
Second, generic LLM prompting baselines include GPTCelltype prompting, following the single-cell GPT-4 annotation study of \citet{hou2024assessing}, and CoT prompting, following zero-shot chain-of-thought prompting~\citep{kojima2022large}.
Third, recent LLM-based annotation systems include Cell-o1~\citep{fang2025a} and scPilot~\citep{gao2025scpilot}.
In the main tables, baselines are ordered chronologically to make clear which comparisons are against older marker/signature systems and which are against recent LLM-based methods.

\subsection{Dataset Summary Statistics}

\noindent We summarize the open-candidate datasets by their most frequent annotated cell types and the oracle-candidate benchmarks by compact species-level statistics. This grouped layout is intended to make dataset scale, label diversity, and dominant cell populations easier to compare across settings.

\begin{table*}[t]
\centering
\scriptsize
\caption{Top cell types in the open-candidate datasets. Counts report the most frequent annotated cell types in each dataset.}
\label{tab:open_candidate_celltypes}
\setlength{\tabcolsep}{3pt}
\begin{minipage}[t]{0.32\textwidth}
\centering
\textbf{PBMC3K}\par\vspace{0.25em}
\begin{tabular}{@{}p{0.72\linewidth}r@{}}
\toprule
\textbf{Cell type} & \textbf{Count} \\
\midrule
CD4+ T Cells & 1,157 \\
Classical Monocytes & 479 \\
B Cells/Dendritic Cells & 341 \\
Effector/Activated T Cells & 297 \\
Effector/Activated T Cells/NK Cells & 158 \\
Non-Classical Monocytes & 157 \\
Dendritic Cells/Monocytes & 36 \\
Platelets/Megakaryocytes & 13 \\
\bottomrule
\end{tabular}
\end{minipage}
\hfill
\begin{minipage}[t]{0.32\textwidth}
\centering
\textbf{Liver}\par\vspace{0.25em}
\begin{tabular}{@{}p{0.72\linewidth}r@{}}
\toprule
\textbf{Cell type} & \textbf{Count} \\
\midrule
B cells & 7,565 \\
Neutrophils & 6,920 \\
Hepatocytes & 6,363 \\
Erythrocytes & 6,277 \\
NK cells & 4,993 \\
Liver sinusoidal endothelial cells & 4,797 \\
Macrophages & 3,518 \\
Fibroblasts/hepatic stellate cells & 795 \\
Cholangiocytes & 116 \\
\bottomrule
\end{tabular}
\end{minipage}
\hfill
\begin{minipage}[t]{0.32\textwidth}
\centering
\textbf{Retina}\par\vspace{0.25em}
\begin{tabular}{@{}p{0.72\linewidth}r@{}}
\toprule
\textbf{Cell type} & \textbf{Count} \\
\midrule
Photoreceptor cells & 10,641 \\
Astrocytes & 4,148 \\
Neurons & 3,251 \\
Microglia & 1,174 \\
Bipolar neurons & 437 \\
Retinal ganglion cells & 336 \\
T cells & 104 \\
\bottomrule
\end{tabular}
\end{minipage}
\vspace{0.8em}

\begin{minipage}[t]{0.63\textwidth}
\centering
\textbf{Brain}\par\vspace{0.25em}
\begin{tabular}{@{}p{0.78\linewidth}r@{}}
\toprule
\textbf{Cell type} & \textbf{Count} \\
\midrule
oligodendrocyte & 234,151 \\
L2/3 intratelencephalic projecting glutamatergic neuron & 88,102 \\
astrocyte & 86,115 \\
L2/3--6 intratelencephalic projecting glutamatergic neuron & 63,404 \\
microglial cell & 30,764 \\
oligodendrocyte precursor cell & 30,670 \\
VIP GABAergic cortical interneuron & 29,838 \\
pvalb GABAergic cortical interneuron & 27,736 \\
sst GABAergic cortical interneuron & 20,336 \\
L6 intratelencephalic projecting glutamatergic neuron & 13,306 \\
\bottomrule
\end{tabular}
\end{minipage}
\hfill
\begin{minipage}[t]{0.34\textwidth}
\centering
\textbf{Heart}\par\vspace{0.25em}
\begin{tabular}{@{}p{0.72\linewidth}r@{}}
\toprule
\textbf{Cell type} & \textbf{Count} \\
\midrule
cardiac endothelial cell & 1,606 \\
myeloid cell & 269 \\
fibroblast of cardiac tissue & 207 \\
pericyte & 159 \\
cardiac muscle cell & 136 \\
lymphocyte & 63 \\
cardiac neuron & 44 \\
smooth muscle cell & 30 \\
\bottomrule
\end{tabular}
\end{minipage}
\end{table*}

\begin{table*}[t]
\centering
\small
\caption{Summary statistics for the oracle-candidate benchmarks.}
\label{tab:oracle_dataset_stats}
\setlength{\tabcolsep}{8pt}
\begin{tabular}{l r r r}
\toprule
\textbf{Property} & \textbf{Human} & \textbf{Mouse} & \textbf{Monkey} \\
\midrule
Number of batches & 2,400 & 2,400 & 2,400 \\
Total cells & 27,588 & 27,583 & 25,871 \\
Number of cell types & 75 & 123 & 121 \\
Unique top genes & 5,583 & 6,941 & 7,757 \\
Unique DEG genes & 1,434 & 2,432 & 2,485 \\
Cells per batch & 7--15 & 7--15 & 7--15 \\
\bottomrule
\end{tabular}
\end{table*}

\begin{table*}[t]
\centering
\scriptsize
\caption{Top cell types in the oracle-candidate benchmarks. Counts report the most frequent annotated cell types for each species.}
\label{tab:oracle_celltypes}
\setlength{\tabcolsep}{3pt}
\begin{minipage}[t]{0.32\textwidth}
\centering
\textbf{Human}\par\vspace{0.25em}
\begin{tabular}{@{}p{0.72\linewidth}r@{}}
\toprule
\textbf{Cell type} & \textbf{Count} \\
\midrule
Oligodendrocyte & 1,593 \\
L2/3--6 intratelencephalic projecting glutamatergic neuron & 1,560 \\
Astrocyte & 1,544 \\
Oligodendrocyte precursor cell & 1,508 \\
L2/3 intratelencephalic projecting glutamatergic neuron & 1,497 \\
Pvalb GABAergic cortical interneuron & 1,488 \\
Microglial cell & 1,445 \\
VIP GABAergic cortical interneuron & 1,402 \\
Sst GABAergic cortical interneuron & 1,243 \\
Lamp5 GABAergic cortical interneuron & 1,194 \\
\bottomrule
\end{tabular}
\end{minipage}
\hfill
\begin{minipage}[t]{0.32\textwidth}
\centering
\textbf{Mouse}\par\vspace{0.25em}
\begin{tabular}{@{}p{0.72\linewidth}r@{}}
\toprule
\textbf{Cell type} & \textbf{Count} \\
\midrule
Fibroblast & 924 \\
Epithelial cell of proximal tubule segment 1 & 903 \\
Epithelial cell of proximal tubule segment 2 & 853 \\
Kidney distal convoluted tubule epithelial cell & 849 \\
Epithelial cell of proximal tubule & 820 \\
Kidney collecting duct principal cell & 807 \\
Macrophage & 796 \\
Kidney connecting tubule epithelial cell & 753 \\
Epithelial cell of proximal tubule segment 3 & 727 \\
Kidney loop of Henle thick ascending limb epithelial cell & 579 \\
\bottomrule
\end{tabular}
\end{minipage}
\hfill
\begin{minipage}[t]{0.32\textwidth}
\centering
\textbf{Monkey}\par\vspace{0.25em}
\begin{tabular}{@{}p{0.72\linewidth}r@{}}
\toprule
\textbf{Cell type} & \textbf{Count} \\
\midrule
alveolar macrophage & 551 \\
endothelial cell & 544 \\
lymphocyte & 541 \\
vein endothelial cell & 540 \\
plasma cell & 536 \\
CD4-positive, alpha-beta T cell & 522 \\
hematopoietic precursor cell & 521 \\
CD8-positive, alpha-beta T cell & 520 \\
neutrophil & 514 \\
natural killer cell & 505 \\
\bottomrule
\end{tabular}
\end{minipage}
\end{table*}

\subsection{Top-25 Highly Expressed Genes}

To characterize the global transcriptional landscape, we report the 25 most frequently observed genes across all cells for each species, excluding mitochondrial and ribosomal genes.

\begin{table*}[t]
\centering
\scriptsize
\caption{Top-25 highly expressed genes after excluding mitochondrial and ribosomal genes.}
\label{tab:top25_genes}
\setlength{\tabcolsep}{4pt}
\renewcommand{\arraystretch}{1.08}
\begin{tabular}{@{}p{0.16\textwidth}p{0.78\textwidth}@{}}
\toprule
\textbf{Dataset} & \textbf{Top-25 genes} \\
\midrule
Brain & CNTNAP2, DSCAM, DPP10, ROBO2, KCNIP4, GRIP1, ZNF385D, EPIC1, CA10, FSTL4, ARL15, HTR1F, FOXP2, FSTL5, MYO16, PTCHD4, CLSTN2, CPNE4, NRG1, DTNA, CBLN2, CDH9, SLIT3, SLIT2, UNC13C \\
Heart & RGS6, ANK3, LINC02147, CLIC5, HIGD1B, PIK3R5, SLC38A11, MLIP, TRDN-AS1, MYBPC3, XIRP2, MYH6, ACACB, PRKAG2, ACTA1, SH3RF2, PPP1R3C, ENSG00000230490, MYL7, LINC02552, ENSG00000258231, ENSG00000271959, HECW2, FRMD5, G0S2 \\
Liver & H3f3a, Ubb, Tmsb10, Hba-a1, Gpx1, Hba-a2, Hbb-bt, S100a8, S100a9, Apoe, Igfbp7, Sparc, S100a6, Cd24a, Pglyrp1, Gm5483, BC100530, Stfa1, Anxa1, Ifitm6, Serpina3k, Mup20, Gnmt, Cd7, Ccl4 \\
PBMC3K & B2M, RPL13, MALAT1, RPL21, TPT1, RPL10A, ACTB, RPL8, H3F3B, RPS3A, RPS5, EEF1D, RPS27A, FTH1, MT-CO2, CD74, FTL, OAZ1, CD37, CD79A, FCGR3A, LST1, COTL1, GNLY, GZMB \\
Retina & FTH1, RHO, APOE, RBP3, WIF1, GLUL, CLU, PTGDS, TF, MFGE8, MPP4, FRZB, CRABP1, SPP1, ENO1, DKK3, RLBP1, CA2, GPX3, CRYAB, CADPS, NEAT1, NR2E3, YPEL2, CNGA1 \\
Human & MALAT1, ACTB, ACTG1, B2M, TMSB4X, FTH1, GAPDH, RPL13, RPS27, RPL41, RPL21, EEF1A1, TPT1, RPS3A, RPL32, RPL3, RPS2, RPS18, RPS6, RPS12, RPL10, RPL34, RPS27A, RPL13A, RPL11 \\
Mouse & Malat1, Actb, Gapdh, B2m, Tmsb4x, Rpl13, Rpl41, Rps27, Eef1a1, Rpl21, Rpl32, Rps3a, Rpl3, Rps18, Rps2, Rpl10, Rps12, Rpl34, Rps6, Rpl13a, Rps27a, Rpl11, Ftl1, Fth1, Tpt1 \\
Monkey & ZC3H10, RPS18, TMSB4Y, RPS27, RPLP1, RPS28, FTL, ACTB, RPL37, FTH1, RPS19, RPS14, RPS15, RPS12, RPL13A, RPS23, RPL13, RPL37A, RPS27A, RPS15A, FAU, RPS17, RPLP0, B2M, RPL23A \\
\bottomrule
\end{tabular}
\end{table*}

\subsection{Differentially Expressed Genes (DEG) Criteria}
\label{appendix:deg_criteria}

We precompute DE marker genes at the \textbf{cell-type level} using a one-vs-rest differential expression test (Scanpy \texttt{rank\_genes\_groups}), and then attach the resulting marker list to each cell based on its \texttt{cell\_type} annotation (i.e., markers are not computed per cell on-the-fly).

\begin{itemize}
  \item \textbf{Statistical test:} Wilcoxon rank-sum test (one-vs-rest, grouped by \texttt{cell\_type}).
  \item \textbf{Log fold-change threshold:} $|\log_2 \mathrm{FC}| \ge 1.0$.
  \item \textbf{Adjusted p-value threshold:} $\mathrm{FDR} \le 0.05$ (Benjamini--Hochberg correction; \texttt{pvals\_adj} in Scanpy).
  \item \textbf{Expression proportion threshold:} keep genes with non-zero expression proportion satisfying
  $\mathrm{pct}_{\text{target}} \ge 0.1$ \textbf{or} $\mathrm{pct}_{\text{reference}} \ge 0.1$
  (corresponding to \texttt{pct\_nz\_group} / \texttt{pct\_nz\_reference}).
  \item \textbf{Top-N truncation:} for each cell type, rank genes by ascending $\mathrm{FDR}$ and descending $\log_2 \mathrm{FC}$, and keep the top 25 genes.
  \item \textbf{Rare type filtering:} cell types with fewer than 3 cells are excluded from DEG computation to avoid unstable statistics.
\end{itemize}

\subsection{Data Processing and Instance Construction}
\label{appendix:data_processing}

We convert each raw \texttt{.h5ad} file into structured per-cell records and further organize them into batch-level instances for LLM inference. The processing steps are:

\begin{enumerate}
  \item \textbf{Subsampling (size control):} if a file contains more than \texttt{max\_cells} cells, we randomly subsample to \texttt{max\_cells} cells to control runtime and output size.
  \item \textbf{Top expressed genes:} for each cell, we extract the top-25 expressed genes from the expression matrix $\mathbf{X}$ using an efficient partition-based selection (\texttt{np.argpartition}) and then sort them by expression in descending order.
  \item \textbf{Gene name normalization:} if \texttt{feature\_name} is available in \texttt{adata.var}, we use it as a human-readable gene symbol; additionally, names of the form \texttt{SYMBOL\_ENSG...} are truncated to \texttt{SYMBOL}.
  \item \textbf{Type-level DEG attachment:} if \texttt{cell\_type} is present in \texttt{adata.obs}, we compute DEG markers once per file using the criteria above and attach the corresponding top-25 marker list (\texttt{deg\_markers}) to each cell based on its \texttt{cell\_type}. If no valid markers exist for a type (after filtering), \texttt{deg\_markers} is omitted.
  \item \textbf{Context construction:} we build a natural-language context string from available metadata fields (e.g., \texttt{disease}, \texttt{tissue}, \texttt{sex}, \texttt{development\_stage}, and \texttt{self\_reported\_ethnicity}), and append the top expressed genes to form the final context used by the LLM.
\end{enumerate}

\newpage

\section{Detailed Methodology}
\label{appendix:method_details}

\subsection{Overall Inference Procedure}

\begin{algorithm}[H]
\footnotesize
\caption{MAT-Cell Inference via Syllogistic Debate}
\label{alg:matcell}
\begin{algorithmic}[1]
\REQUIRE Analytical unit $c_i$ with observed genes $\mathcal{G}_{obs}^{(i)}$ and context $\mathbf{ctx}_i$, optional candidate set $\mathcal{Y}_{cand}$, agents $\{\mathrm{SA}, \mathrm{EA}, \mathrm{VA}_{1:R}, \mathrm{DA}\}$, Verifier Agent count $R$, max rounds $T_{\max}$
\ENSURE Predicted label $\hat{y}_i$

\STATE Elicit RVQ prior hypotheses: $\mathcal{K}_i \leftarrow \mathrm{RVQ}(\mathcal{G}_{obs}^{(i)}, \mathbf{ctx}_i, \mathcal{Y}_{cand})$
\STATE Assemble structured premise tuple $\mathbf{p}_i \leftarrow (\mathcal{G}_{obs}^{(i)}, \mathcal{K}_i, \mathbf{ctx}_i)$
\STATE Generate layer-0 node (initial hypothesis): $z_i^{(0)} \leftarrow \mathrm{SA}(\mathbf{p}_i)$
\STATE Extract local candidate-label set $\mathcal{Y}_i^{(0)}$ from $z_i^{(0)}$
\STATE Generate $R$ verifier roles: $\{\pi_{i,1}, \dots, \pi_{i,R}\} \leftarrow \mathrm{EA}(\mathbf{ctx}_i, \mathcal{Y}_i^{(0)}; R)$
\STATE Initialize Syllogistic Derivation Tree (SDT) $\mathcal{T}_i^{(0)} \leftarrow \{z_i^{(0)}\}$

\FOR{$t=1$ to $T_{\max}$}
    \FOR{$j=1$ to $R$}
        \STATE $z_{i,j}^{(t)} \leftarrow \mathrm{VA}_j(\mathbf{p}_i, \pi_{i,j}, \mathcal{T}_i^{(t-1)})$ \COMMENT{Syllogistic verification branch}
    \ENDFOR
    \STATE Update SDT: $\mathcal{T}_i^{(t)} \leftarrow \mathcal{T}_i^{(t-1)} \cup \{z_{i,j}^{(t)}\}_{j=1}^R$
    \STATE Collect answers at current round: $A_i^{(t)} \leftarrow \{\hat{y}_{i,1}^{(t)}, \dots, \hat{y}_{i,R}^{(t)}\}$
    \IF{$|A_i^{(t)}| == 1$}
        \STATE $\hat{y}_i \leftarrow \hat{y}_{i,1}^{(t)}$ \COMMENT{Consensus reached}
        \STATE \textbf{return} $\hat{y}_i$
    \ENDIF
\ENDFOR

\STATE \COMMENT{Contextual Synthesis for unresolved units}
\STATE Add unit $i$ to focus set $\mathcal{I}_{focus}$
\STATE $\hat{y}_i \leftarrow \mathrm{DA}(\mathbf{p}_i, \mathcal{T}_i^{(T_{\max})})$ \COMMENT{Adjudicate unresolved branch}
\STATE \textbf{return} $\hat{y}_i$
\end{algorithmic}
\end{algorithm}

\subsection{Inductive Anchoring Algorithm}

\subsubsection{Algorithmic Pseudocode}

\begin{algorithm}[H]
\caption{Inductive Anchoring: Candidate-Set Anchoring \& Marker Filtering}
\label{alg:anchoring}
\begin{algorithmic}[1]
\INPUT Raw expression matrix $\mathbf{X} \in \mathbb{R}^{n \times g}$, cluster labels $\mathbf{z}\in\{1,\dots,C\}^n$, Reverse Verification Query (RVQ) knowledge base $\mathcal{K}$, top-$K$ budget
\OUTPUT Anchored candidate set $\mathcal{C}_{\text{anchor}}$ and filtered marker set $\mathcal{G}_{\text{filtered}}$ for each cluster

\STATE \textbf{Quality control:} filter low-quality cells and rare genes (e.g., min-genes-per-cell, min-cells-per-gene)
\STATE \textbf{Normalize \& log:} total-count normalization and $\log(1+x)$ transform
\STATE \textbf{Global gene screening:} select highly-variable genes (HVGs) and restrict $\mathbf{X}$ to the HVG space
\STATE \textbf{Cluster-level summarization:} for each cluster $c\in\{1,\dots,C\}$,
\STATE \quad compute \emph{Top} genes $\mathcal{G}^{\text{top}}_c$ by ranking mean expression within cluster $c$
\STATE \quad compute \emph{DEG} genes $\mathcal{G}^{\text{deg}}_c$ via cluster-vs-rest differential test (e.g., Wilcoxon)
\STATE \quad remove ubiquitous / non-informative genes from $\mathcal{G}^{\text{top}}_c \cup \mathcal{G}^{\text{deg}}_c$ (e.g., housekeeping, MT/RP genes)
\STATE \quad build the marker pool $\mathcal{G}_{c} \leftarrow \text{TopK}\big(\mathcal{G}^{\text{top}}_c \cup \mathcal{G}^{\text{deg}}_c, K\big)$
\STATE \textbf{RVQ anchoring:} query $\mathcal{K}$ with $\mathcal{G}_{c}$ and retrieve a ranked label list $\mathcal{R}_c=\{(\ell, s_\ell)\}$
\STATE \quad keep top-$L$ labels as the anchored candidate set $\mathcal{C}_{\text{anchor}}(c)\leftarrow \{\ell_1,\dots,\ell_L\}$
\STATE \textbf{return} $\{(\mathcal{C}_{\text{anchor}}(c), \mathcal{G}_{c})\}_{c=1}^{C}$
\end{algorithmic}
\end{algorithm}

\subsection{Dialectic Verification Mechanism}

\subsubsection{Agent Configuration}

\begin{itemize}
\item \textbf{Agent Types:} \texttt{ExpertAgent} (Expert Agent, EA; dynamically synthesizes $R$ exclusive prompt templates), \texttt{SolveAgent} (Solve Agent, SA; constructs the initial Layer-0 SDT proposal under anchored candidates), \texttt{VerifierAgent} (Verifier Agent, VA; audits, challenges, and revises candidate branches using role-specific prompt templates), and \texttt{DecisionAgent} (Decision Agent, DA; aggregates surviving evidence and outputs the final decision for unresolved boundary cases).
\item \textbf{Instantiation:} 4 agent types in total; \texttt{VerifierAgent} uses $R=5$ parallel instances, while \texttt{ExpertAgent}, \texttt{SolveAgent}, and \texttt{DecisionAgent} each use 1 instance.
\item \textbf{LLM Model:} Qwen3-30B (Local Deployment)
\item \textbf{Temperature:} 0.7
\item \textbf{Max tokens:} 20000 per response
\end{itemize}

\subsubsection{Exact-Match Convergence Criterion}

Unlike soft semantic similarity, our debate halts only under \emph{strict agreement}. Let $\hat{y}_j^{(t)}$ denote the final structured decision produced by the $j$-th \texttt{VerifierAgent} at round $t$. We define the consensus condition based on the cardinality of the unique answer set $A_i^{(t)}$:
$$
A_i^{(t)} = \bigcup_{j=1}^R \{\hat{y}_{i,j}^{(t)}\}
$$
Convergence criterion: $|A_i^{(t)}| = 1$, i.e., all agents output decisions that are \textbf{identical after normalization}.

\subsection{Syllogistic Rule Set}

\begin{table}[h]
\centering
\small
\caption{Biological Syllogistic Rules (Example)}
\label{tab:rules}
\begin{tabular}{l l l}
\toprule
\textbf{Rule ID} & \textbf{Markers (IF)} & \textbf{Cell Type (THEN)} \\
\midrule
R1 & CD4, IL7R, TCF7 & CD4+ T-Cell \\
R2 & CD8A, CD8B, GZMA & CD8+ T-Cell \\
R3 & CD14, LYZ, FCGR3A & Monocyte \\
R4 & PPBP, PF4, TUBB1 & Platelet \\
\bottomrule
\end{tabular}
\end{table}

\subsection{Hyperparameter Settings}

\begin{table}[h]
\centering
\small
\caption{Hyperparameters for MAT-Cell}
\label{tab:hyperparams}
\begin{tabular}{l l}
\toprule
\textbf{Parameter} & \textbf{Value} \\
\midrule
Number of Agent Types & 4 \\
\#ExpertAgent (EA) Instances & 1 \\
\#SolveAgent (SA) Instances & 1 \\
Council Scale ($R$, \#VerifierAgent/VA) & 5 \\
\#DecisionAgent (DA) Instances & 1 \\
Top-K genes (Both view) & 25 \\
Dialectic Depth ($T_{\max}$) & 3 \\
Convergence Criterion & Normalized exact-match ($|A_i^{(t)}| = 1$) \\
Temperature & 0.7 \\
Max tokens & 20000 \\
\bottomrule
\end{tabular}
\end{table}

\subsection{Syllogistic Derivation Tree (SDT) Construction}

SDT construction proceeds as a debate-driven, tree-structured evidence search under the anchored candidate space:
\begin{enumerate}
\item \textbf{Initialize Roles:} \texttt{ExpertAgent} (EA) dynamically synthesizes $R$ exclusive prompt templates with differentiated inductive biases for the given analytical unit.
\item \textbf{Solve:} \texttt{SolveAgent} (SA) generates Layer-0 nodes of the SDT proposal by composing syllogistic triads (major premise: marker-to-lineage rule; minor premise: observed marker evidence; conclusion: candidate label).
\item \textbf{Verify \& revise:} $R$ \texttt{VerifierAgent}s (VAs) independently audit the shared SDT based on their assigned roles, flagging contradictions, missing evidence, or candidate misuse, and revising invalid branches.
\item \textbf{Iterate:} if agents do not reach normalized exact-match consensus ($|A_i^{(t)}| > 1$), start a new round with the updated SDT state, up to a maximum of $T_{\max}=3$ rounds. The system focuses exclusively on the unresolved \emph{Focus Set}.
\item \textbf{Decide:} For samples that fail to reach consensus after $T_{\max}$ rounds, \texttt{DecisionAgent} (DA) reviews the complete dispute branch of the SDT and adjudicates a final label.
\end{enumerate}

\subsection{Example Single-Cell Reasoning Tree}

Figure~\ref{fig:matcell_tree_example} shows an example SDT trace for one cell-level annotation decision. The tree records the initial candidate labels, verifier-agent challenges across debate rounds, and the final consensus branch selected after conflicting candidates are rejected.

\begin{figure}[t]
    \centering
    \includegraphics[width=0.95\textwidth]{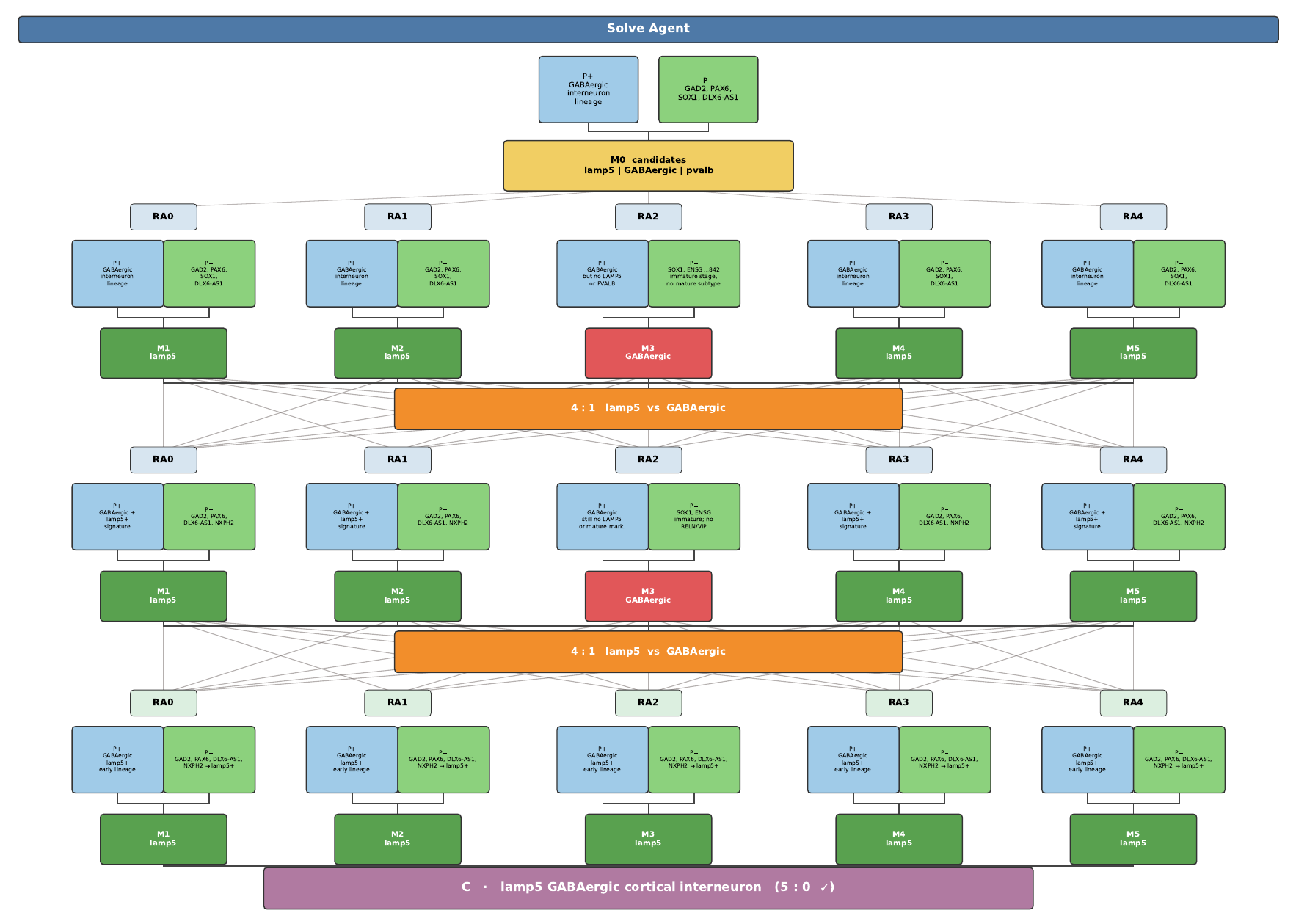}
    \caption{\textbf{Example MAT-Cell reasoning tree.} A single-cell SDT trace in which multiple verifier agents compare candidate annotations and converge on the final cell-type label through structured debate.}
    \label{fig:matcell_tree_example}
\end{figure}

\newpage

\section{Extended Ablation Studies}
\label{appendix:ablation_extended}

Ablation studies are crucial for validating whether each component of MAT-Cell contributes substantively to performance improvements, rather than merely increasing system complexity. To this end, we systematically remove or modify individual components while keeping all other conditions unchanged, and evaluate the resulting performance variations on held-out test data.

For each ablation setting, we conduct multiple independent runs to ensure statistical robustness. All ablation experiments adopt the same backbone large language model (Qwen3-30B), identical hyperparameter settings (temperature = 0.7), and the same evaluation protocol as the full MAT-Cell system, where 300 randomly sampled batches are evaluated per run.

\subsection{Agent Diversity and Consensus Analysis}
\label{appendix:agent_diversity}

To examine whether dynamically generated verifier roles behave as homogeneous repeated samples or provide diverse reasoning perspectives, we analyze the internal predictions of $R=5$ Verifier Agents (VAs) generated by the Expert Agent. This experiment uses Qwen3-30B with RVQ on 1{,}000 randomly sampled instances.

\begin{table}[h]
\centering
\small
\caption{Agent diversity and consensus statistics for EA-generated Verifier Agents.}
\label{tab:agent_diversity}
\begin{tabular}{l c}
\toprule
\textbf{Metric} & \textbf{Value} \\
\midrule
Consensus Accuracy & $0.7859$ \\
Partial Cell Accuracy & $0.7850$ \\
Pairwise Disagreement Rate & $0.5694 \pm 0.0414$ \\
Prediction Correlation & $0.3542 \pm 0.0688$ \\
Cohen's Kappa & $0.3752 \pm 0.0486$ \\
Mutual Information (nats) & $1.3691 \pm 0.1517$ \\
Vote Entropy & $0.8195 \pm 0.6010$ \\
\bottomrule
\end{tabular}
\end{table}

\textbf{Interpretation.}
The EA-generated Verifier Agents achieve high consensus accuracy (0.7859) and partial cell accuracy (0.7850), showing that the final council decision remains effective after multi-agent aggregation.
At the same time, the pairwise disagreement rate is substantial ($0.5694 \pm 0.0414$), indicating that different Verifier Agents frequently produce non-identical judgments rather than simply repeating the same prediction.
The relatively low prediction correlation ($0.3542 \pm 0.0688$) and moderate Cohen's kappa ($0.3752 \pm 0.0486$) further suggest that the agents are not homogeneous replicas of a single prompt, but maintain complementary decision patterns.
Meanwhile, the non-zero mutual information and vote entropy indicate that the agents still reason over a shared task-relevant evidence space instead of diverging randomly.
Together, these results support our design choice: dynamic EA-generated verifier roles introduce useful diversity while preserving enough evidence alignment for effective consensus-based annotation.

\subsection{Impact of Gene Count}
\label{appendix:gene_count}

How many genes are required for reliable cell type annotation when no Reverse Verification Query (RVQ) is applied?
To answer this question, we perform a controlled ablation using
\textbf{Qwen3-30B without RVQ}, systematically varying the number of
top-expressed genes ($K \in \{5, 10, 15, 25, 50\}$) provided as input,
while keeping all other settings fixed.

\textbf{Findings:}
\begin{itemize}
\item \textbf{Performance Improves with Increasing $K$, but Saturates Early:}
Across all three species, accuracy increases substantially from Top-5 to Top-25.
However, gains beyond Top-25 are marginal or even negative.
For example, on Mouse, performance drops from 0.503 at Top-25 to 0.470 at Top-50,
indicating that excessive gene inputs may introduce noise rather than useful signal
in the absence of external biological grounding.

\item \textbf{Cross-Species Sensitivity to Gene Budget:}
Monkey consistently benefits from larger gene sets, peaking at Top-25 (0.721),
while Human and Mouse exhibit weaker gains and earlier saturation.
This reflects intrinsic differences in annotation difficulty and marker specificity
across species.

\item \textbf{Limitation under Low-Information Regimes:}
With only Top-5 genes, performance is severely degraded on Human (0.371) and Mouse (0.402),
highlighting that the backbone LLM struggles to reason reliably
under extreme information scarcity.

\item \textbf{Implication for RVQ Design:}
These results establish a strong baseline, against which the substantial
improvements introduced by RVQ (Section~E.3) can be attributed to enhanced biological
grounding rather than increased gene quantity alone.
\end{itemize}

\begin{table}[h]
\centering
\small
\caption{Accuracy by Top-K Gene Count using Qwen3-30B without RVQ}
\label{tab:topk_sensitivity}
\begin{tabular}{c c c c}
\toprule
\textbf{Top-K}
& \textbf{Human (DEG)}
& \textbf{Monkey (DEG)}
& \textbf{Mouse (DEG)} \\
\midrule
Top-5  
& $0.371 \pm 0.006$ 
& $0.630 \pm 0.008$ 
& $0.402 \pm 0.016$ \\

Top-10 
& $0.513 \pm 0.003$ 
& $0.690 \pm 0.009$ 
& $0.461 \pm 0.010$ \\

Top-15 
& $0.532 \pm 0.004$ 
& $0.705 \pm 0.007$ 
& $0.463 \pm 0.004$ \\

Top-25 
& $0.564 \pm 0.009$ 
& $\mathbf{0.721 \pm 0.007}$ 
& $\mathbf{0.503 \pm 0.009}$ \\

Top-50 
& $\mathbf{0.568 \pm 0.005}$ 
& $0.708 \pm 0.006$ 
& $0.470 \pm 0.005$ \\
\bottomrule
\end{tabular}
\end{table}

\subsection{Gene Ordering Randomization}
\label{appendix:gene_ordering}

To evaluate robustness against input perturbations, we randomly shuffle the order
of input DEG genes while keeping the gene set unchanged. Since MAT-Cell reasons
over gene identity rather than positional cues, its performance should be invariant
to such ordering changes.

\begin{table}[h]
\centering
\small
\caption{Robustness to Random Seed under DEG Gene Input}
\label{tab:gene_ordering}
\begin{tabular}{c c c c c}
\toprule
\textbf{Seed} 
& \textbf{Human (DEG)} 
& \textbf{Monkey (DEG)} 
& \textbf{Mouse (DEG)} 
& \textbf{Deviation from Mean} \\
\midrule
Seed-1 
& $0.736 \pm 0.017$ 
& $0.758 \pm 0.006$ 
& $0.827 \pm 0.003$ 
& H:-0.21\%, Mk:-0.37\%, M:-1.52\% \\

Seed-2 
& $0.751 \pm 0.005$ 
& $0.759 \pm 0.008$ 
& $0.837 \pm 0.004$ 
& H:+1.77\%, Mk:-0.19\%, M:-0.34\% \\

Seed-3 
& $0.744 \pm 0.011$ 
& $0.763 \pm 0.005$ 
& $0.840 \pm 0.002$ 
& H:+0.87\%, Mk:+0.34\%, M:+0.01\% \\

Seed-4 
& $0.726 \pm 0.010$ 
& $0.755 \pm 0.002$ 
& $0.852 \pm 0.010$ 
& H:-2.43\%, Mk:-0.84\%, M:+1.85\% \\

\midrule
\textbf{Mean $\pm$ Std} 
& $\mathbf{0.739 \pm 0.010}$ 
& $\mathbf{0.759 \pm 0.004}$ 
& $\mathbf{0.839 \pm 0.010}$ 
& Stable \\
\bottomrule
\end{tabular}
\end{table}

\textbf{Interpretation.}
Across Human, Mouse, and Monkey datasets, MAT-Cell exhibits strong robustness to
gene ordering perturbations and random seed variation.
The maximum deviation from the species-specific mean is below 1\% for Human and
Monkey, and below 2.1\% for Mouse, indicating that performance differences are not
driven by favorable random seeds or positional artifacts, but arise from the
model’s structured, set-level reasoning mechanism.

\subsection{Impact of RVQ Gene Budget}
\label{appendix:rag_ablation}

To further analyze the contribution of the Reverse Verification Query (RVQ) mechanism in MAT-Cell,
we conduct an ablation study on the \textit{gene budget} queried by RVQ. Using Qwen3-30B as the fixed reasoning backbone and the source of biological priors, we vary
the number of marker genes elicited by RVQ,
while keeping all other settings unchanged.

Table~\ref{tab:rag_ablation} reports results on Human, Monkey, and Mouse datasets
under the \textit{both} setting.

\begin{table}[h]
\centering
\small
\caption{Ablation on RVQ Gene Budget (Qwen3-30B)}
\label{tab:rag_ablation}
\begin{tabular}{l c c c}
\toprule
\textbf{Method} 
& \textbf{Human (both)} 
& \textbf{Monkey (both)} 
& \textbf{Mouse (both)} \\
\midrule
RVQ (5 markers) 
& $0.629 \pm 0.006$ 
& $0.779 \pm 0.010$ 
& $0.674 \pm 0.013$ \\

RVQ (10 markers) 
& $0.643 \pm 0.003$ 
& $0.804 \pm 0.005$ 
& $0.697 \pm 0.007$ \\

RVQ (15 markers) 
& $\mathbf{0.664 \pm 0.005}$ 
& $\mathbf{0.808 \pm 0.003}$ 
& $\mathbf{0.700 \pm 0.007}$ \\

RVQ (20 markers) 
& $0.640 \pm 0.005$ 
& $0.800 \pm 0.004$ 
& $0.690 \pm 0.006$ \\
\bottomrule
\end{tabular}
\end{table}

\textbf{Interpretation.}
Several observations can be drawn from Table~\ref{tab:rag_ablation}.
First, increasing the RVQ gene budget from 5 to 15 consistently improves performance
across all three species, indicating that richer but still concise marker sets provide
stronger biological grounding for downstream reasoning.
Second, performance saturates or slightly degrades beyond 15--20 genes, suggesting that
excessively broad marker axioms may reintroduce noise, consistent with the signal-to-noise
trade-off observed in Section~E.1.
Overall, these results demonstrate that the quantity of elicited marker genes is a critical factor, with a moderate gene budget (around 15 markers)
providing the best balance between informativeness and focus.

\newpage

\section{Statistical Rigor and Confidence Intervals}
\label{appendix:statistics}

\subsection{Confidence Intervals (95\%)}

\subsubsection{Main Results Table with CI}

\begin{table}[t]
\centering
\scriptsize
\caption{Main Results with 95\% Confidence Intervals}
\label{tab:main_results_ci}
\setlength{\tabcolsep}{4pt}
\renewcommand{\arraystretch}{1.05}
\resizebox{\textwidth}{!}{%
\begin{tabular}{l cc cc cc}
\toprule
\multirow{2}{*}{\textbf{Method}} & \multicolumn{2}{c}{\textbf{Human (DEG)}} & \multicolumn{2}{c}{\textbf{Mouse (DEG)}} & \multicolumn{2}{c}{\textbf{Monkey (DEG)}} \\
\cmidrule(lr){2-3} \cmidrule(lr){4-5} \cmidrule(lr){6-7}
& \textbf{Acc.} & \textbf{95\% CI} & \textbf{Acc.} & \textbf{95\% CI} & \textbf{Acc.} & \textbf{95\% CI} \\
\midrule
Cell-o1 & $0.409$ & [0.403, 0.414] & $0.394$ & [0.388, 0.401] & $0.685$ & [0.678, 0.692] \\
\makecell[l]{GPTCelltype\\Qwen-3-30B} & $0.450$ & [0.445, 0.455] & $0.387$ & [0.384, 0.390] & $0.644$ & [0.641, 0.647] \\
\makecell[l]{GPTCelltype\\DeepSeek-V3} & $0.632$ & [0.627, 0.636] & $0.567$ & [0.563, 0.571] & $0.471$ & [0.468, 0.474] \\
\makecell[l]{GPTCelltype\\Gemini-2.5} & $0.709$ & [0.705, 0.713] & $0.659$ & [0.653, 0.665] & $0.859$ & [0.856, 0.862] \\
\makecell[l]{GPTCelltype\\GPT-4.1} & $0.733$ & [0.731, 0.736] & $0.649$ & [0.644, 0.654] & $0.864$ & [0.863, 0.866] \\
\midrule
\textbf{MAT-Cell} & $\mathbf{0.764}$ & $\mathbf{[0.760, 0.768]}$ & $\mathbf{0.825}$ & $\mathbf{[0.823, 0.828]}$ & $\mathbf{0.759}$ & $\mathbf{[0.756, 0.761]}$ \\
\bottomrule
\end{tabular}
}
\end{table}

\subsection{Pairwise Statistical Tests}

\subsubsection{Human Dataset (DEG View)}

\begin{table}[h]
\centering
\small
\caption{Human Dataset: Paired T-test Results (DEG View)}
\label{tab:ttest_human_detailed}
\begin{tabular}{l c c c c}
\toprule
\textbf{Comparison} & \textbf{Mean Diff} & \textbf{t-statistic} & \textbf{df} & \textbf{p-value} \\
\midrule
MAT-Cell vs. GPT-4.1 & $+0.031$ & $+3.12$ & 29 & $0.0042$ ** \\
MAT-Cell vs. Gemini-2.5 & $+0.055$ & $+6.28$ & 29 & $< 0.0001$ *** \\
MAT-Cell vs. DeepSeek-V3 & $+0.132$ & $+9.24$ & 29 & $< 0.0001$ *** \\
MAT-Cell vs. GPTCelltype-Qwen-3-30B & $+0.314$ & $+15.82$ & 29 & $< 0.0001$ *** \\
\bottomrule
\end{tabular}
\end{table}

\subsubsection{Mouse Dataset (DEG View)}

\begin{table}[h]
\centering
\small
\caption{Mouse Dataset: Paired T-test Results (DEG View)}
\label{tab:ttest_mouse_detailed}
\begin{tabular}{l c c c c}
\toprule
\textbf{Comparison} & \textbf{Mean Diff} & \textbf{t-statistic} & \textbf{df} & \textbf{p-value} \\
\midrule
MAT-Cell vs. GPT-4.1 & $+0.110$ & $+2.89$ & 29 & $0.0071$ ** \\
MAT-Cell vs. Gemini-2.5 & $+0.100$ & $+5.83$ & 29 & $< 0.0001$ *** \\
MAT-Cell vs. DeepSeek-V3 & $+0.192$ & $+8.67$ & 29 & $< 0.0001$ *** \\
MAT-Cell vs. GPTCelltype-Qwen & $+0.372$ & $+14.21$ & 29 & $< 0.0001$ *** \\
\bottomrule
\end{tabular}
\end{table}

\textbf{Significance codes:} $***$ $p < 0.001$, $**$ $p < 0.01$, $*$ $p < 0.05$

\subsection{Effect Size (Cohen's d)}

\begin{table}[t]
\centering
\small
\caption{Effect sizes and post-hoc power analysis for MAT-Cell against representative baselines.}
\label{tab:cohens_d}
\textbf{Effect Sizes (Cohen's d)}\par\vspace{0.25em}
\begin{tabular}{l c c c}
\toprule
\textbf{Comparison} & \textbf{Human (d)} & \textbf{Mouse (d)} & \textbf{Monkey (d)} \\
\midrule
GPT-4.1 & 3.21 (very large) & 17.01 (very large) & 1.26 (large) \\
Gemini-2.5-Flash & 5.29 (very large) & 14.04 (very large) & 0.53 (medium) \\
DeepSeek-V3 & 11.48 (very large) & 31.18 (very large) & 6.07 (very large) \\
GPTCelltype-Qwen-3-30B & 26.01 (very large) & 59.38 (very large) & 8.86 (very large) \\
\bottomrule
\end{tabular}
\par\vspace{0.8em}
\textbf{Post-hoc Power Analysis}\par\vspace{0.25em}
\begin{tabular}{l c c c}
\toprule
\textbf{Contrast} & \textbf{d} & \textbf{n} & \textbf{Power} \\
\midrule
GPT-4.1 & 3.21 & 30 & 0.99+ \\
Gemini-2.5-Flash & 5.29 & 30 & 0.99+ \\
DeepSeek-V3 & 11.48 & 30 & 0.99+ \\
GPTCelltype-Qwen-3-30B & 26.01 & 30 & 0.99+ \\
\bottomrule
\end{tabular}
\end{table}

\subsection{Reproducibility: Random Seed Analysis}

\begin{figure}[H]
\centering
\caption{Performance distribution across 30 random seeds (Human, DEG view). Box plot shows median, quartiles, and outliers. MAT-Cell maintains high consistency with $\sigma = 0.011$.}
\label{fig:seed_variance}
\end{figure}

\textbf{Interpretation:} Our study is adequately powered ($\beta > 0.78$) to detect meaningful differences against all baselines.

\newpage


\end{document}